\documentclass[letterpaper]{article}

\ifdefined\aaaianonymous
    \usepackage[submission]{aaai2026}
\else
    \usepackage{aaai2026}
\fi

\usepackage{times}
\usepackage{helvet}
\usepackage{courier}
\usepackage[hyphens]{url}
\usepackage{graphicx}
\usepackage{natbib}
\usepackage{caption}
\usepackage{amsmath}
\usepackage{amssymb}
\usepackage{amsthm}

\usepackage{mathtools}
\usepackage{booktabs}
\usepackage{array}
\usepackage{placeins}
\usepackage{float}
\usepackage{multirow}
\newtheorem{theorem}{Theorem}
\newtheorem{corollary}{Corollary}
\newtheorem{proposition}{Proposition}

\theoremstyle{definition}

\let\Cap\undefined  % amssymb defines \Cap as \bigcap; we override
\DeclareMathOperator{\Cap}{Cap}
\DeclareMathOperator{\Rob}{Rob}

\newcommand{\eps}{\varepsilon}
\newcommand{\Astar}{A^{\star}}
\newcommand{\Api}{A_{\pi}}
\newcommand{\Atildepi}{\tilde A_{\pi}}
\newcommand{\Xtilde}{\tilde X}

\title{Capability and Robustness Cannot Both Be Free:\\
An Information-Theoretic Bound for Vision-Language-Action Models}

\author{Jianwei Tai}
\affiliations{School of Internet, Anhui University\\
24012@ahu.edu.cn}

\begin{document}

\maketitle

\begin{abstract}
Vision-Language-Action (VLA) models reach high success rates on clean inputs but collapse under small adversarial perturbations: a $16/255$ PGD attack drops OpenVLA-7B's LIBERO success from $95\%$ to under $5\%$. Whether this trade-off has a theoretical floor was open. We prove that it does. For any VLA policy, capability $I(\Astar;\Api)$ and robustness $I(\Api;\Atildepi)-I(\Api;\delta)$ sum to at most $H(\Astar)+I(X;\Xtilde)$, the task entropy plus adversarial channel capacity. The proof reduces to two applications of the Data Processing Inequality. The pixel-level bound is loose by $\sim 10^3$ nats and serves as a ceiling guarantee; an encoder-specific corollary tightens it by over an order of magnitude, into a regime where realized capability already consumes $5$--$9\%$ of the budget. We validate Theorem~\ref{thm:main} with zero violations across $308$ cells: $252$ closed-form Gaussian-VLA, $48$ OpenVLA-7B$+$LIBERO$+$PGD ($4$ suites $\times$ $4$ $\eps$ $\times$ $3$ seeds), $4$ Square-Attack, and $4$ multi-step ($T{=}10$). A complementary measurability inequality $\Rob_{\text{disc}} \le \Cap_{\text{disc}}$ further holds across $144$ cross-architecture cells spanning OpenVLA, OpenVLA-OFT (continuous-$L_1$), and SmolVLA (flow-matching). The same construction yields three label-free diagnostics: a pre-flight encoder ceiling, a defense-forensics probe that localizes input-side vs.\ language-model intervention, and a head-agnostic robustness ratio comparable across discrete-token, $L_1$-regression, and flow-matching policies. Together these provide the cross-setting axis defense and architecture comparisons currently lack.
\end{abstract}

\ifdefined\aaaianonymous
% \begin{links} suppressed for anonymous submission
\else
%\begin{links}
%    \link{Code}{https://anonymous.example/vla-impossibility}
%\end{links}
\fi

\section{Introduction}
\label{sec:intro}

Vision-Language-Action (VLA) models, large multimodal policies that map an image and a natural-language instruction directly to a robot action, account for much of the past two years' progress in generalist manipulation and are now beginning to operate alongside humans on real hardware. OpenVLA-7B \citep{kim2024openvla} reaches $95.4\%$ success on LIBERO-Spatial; RT-2 \citep{brohan2023rt2} generalizes to unseen object-verb compositions; $\pi_0$ \citep{black2024pi0} scales to dexterous bimanual tasks. Clean-input success rates have climbed steadily year over year. Yet each prediction is a physical action, which turns visual fragility from an accuracy issue into a safety one.

Under attack the picture shifts. Recent empirical studies \citep{robustvla2025,multimodaladv2025,explainadvvla2025} report a consistent pattern: a $16/255$ PGD perturbation imperceptible to humans drops OpenVLA's task success on LIBERO from $95\%$ to under $5\%$. Defenses recover only part of the loss. Adversarial fine-tuning typically restores $10$--$50\%$ of the lost robustness at a cost of $5$--$15\%$ in clean accuracy, and the resulting Pareto frontier is uneven across datasets and architectures. The prior question is whether the trade-off has a theoretical floor at all, or whether better training could deliver clean accuracy and adversarial invariance together without limit. \citet{tsipras2019robustness} settled the analogous question for binary classifiers under specific data distributions; the corresponding statement for action-generating policies has, to our knowledge, never been written down. This paper writes it down.

For any VLA policy $\pi$ producing action $\Api$ from observation $X$, with oracle action $\Astar$ and adversarial perturbation $\delta$ such that $\Xtilde = X + \delta$ and $\Atildepi = \pi(\Xtilde)$:
\begin{equation}
\underbrace{I(\Astar; \Api)}_{\text{Capability}} + \underbrace{I(\Api; \Atildepi) - I(\Api; \delta)}_{\text{Robustness}} \le \underbrace{H(\Astar) + I(X; \Xtilde)}_{\text{Budget}}.
\label{eq:intro-bound}
\end{equation}
The proof reduces to two applications of the Data Processing Inequality plus standard non-negativity facts (Theorem~\ref{thm:main}). The right-hand side depends only on the task and the attacker, so any architectural change must trade capability against robustness within a fixed budget. On current models the bound is loose (slack $\sim 10^3$ nats, dominated by the gap between image-dimensional channel capacity and 7-D action geometry), so it functions as a ceiling guarantee rather than a tight predictor of achievable performance.

The bound is \emph{non-constructive}: it asserts $\Cap + \Rob \le \text{Budget}$ without prescribing a $\pi$ that attains equality. We propose the slack $S(\pi, \eps) = H(\Astar) + I(X;\Xtilde) - \Cap(\pi) - \Rob(\pi) \ge 0$ as a comparable cross-paper metric, allowing a defense paper to report ``we closed $60\%$ of the slack at $\eps{=}4/255$'' rather than absolute robustness numbers that fail to transfer across settings. Our MI estimators are biased; Sec.~\ref{sec:synthetic} catalogs the direction. The bias does not threaten the inequality, since estimator error shrinks the LHS faster than the RHS in our regime, widening the slack rather than producing spurious violations.

The bound also serves as more than a verification target. The same calculation yields three quantities that practitioners can read directly: a $\sim 5$-minute pre-flight encoder ceiling that audits the worst-case channel capacity of a candidate $(\pi, \text{defense}, \eps)$ deployment without running an attack, a defense-forensics shift signature that identifies whether a defense intervenes input-side or LLM-side, and a head-agnostic $\Rob_{\text{disc}}/\Cap_{\text{disc}}$ ratio that places discrete-token, $L_1$-regression, and flow-matching action heads on a common axis. Each is computable from $\le 200$ samples, requires no ground-truth labels, and is comparable across policies, defenses, and architectures.

The contributions of this paper are summarized as follows.

\textbf{1. An information-theoretic ceiling for VLA capability and robustness} (Sec.~\ref{sec:theory}). We prove that for any VLA policy $\pi$, capability (mutual information between policy action and oracle action) and robustness (mutual information preserved under attack, net of trivial channel leakage) sum to at most the task entropy plus the attack channel capacity (Theorem~\ref{thm:main}). The proof reduces to two applications of the Data Processing Inequality and MI non-negativity. Because the budget depends only on the task and the attack, gains in one term cost either the other term or the slack, regardless of policy class. The universal bound is substitutable across defenses but loose by $\sim 10^3$ nats on current models, since image-dimensional channel capacity dwarfs the $7$-D action geometry. An encoder-specific corollary (Cor.~\ref{cor:encoder}) tightens it by $28$--$68\times$ at the cost of becoming policy- and defense-specific (it depends on the realized $\sigma^2_{\delta,\phi}$). Together they separate a universal ceiling from a tight per-experiment diagnostic that underpins the practical tools in Contribution~3.

\textbf{2. Empirical validation across $308$ Theorem-1 cells with zero violations} (Secs.~\ref{sec:synthetic},~\ref{sec:realmodel}). Theorem~\ref{thm:main} is verified on $252$ closed-form Gaussian-VLA cells; $48$ OpenVLA-7B$+$LIBERO PGD cells ($4$ suites $\times$ $4$ $\eps$ $\times$ $3$ seeds, pixel-PCA RHS); $4$ Square-Attack cells (one per LIBERO suite at $\eps{=}8/255$, $200$ queries) that rule out a PGD-specific artifact; and $4$ multi-step DPI cells ($T{=}10$, $30$ trajectories per suite) confirming cumulative slack scales as $T \cdot S_1$. A complementary discrete inequality $\Rob_{\text{disc}} \le \Cap_{\text{disc}}$ further holds across $144$ cross-architecture cells ($3$ architectures $\times$ $4$ LIBERO suites $\times$ $4$ perturbation budgets $\times$ $3$ seeds) spanning OpenVLA, OpenVLA-OFT (continuous-$L_1$), and SmolVLA (flow-matching). The bound holds in every cell, and the encoder-tightened RHS keeps realized capability well within the budget, showing that the constraint is informative at policy-relevant scale.

\textbf{3. A deployment-ready diagnostic toolkit derived from the bound} (Secs.~\ref{sec:defenses},~\ref{sec:cross-model}). The same calculation yields three diagnostics, each computable from $\le 200$ samples on a single GPU without ground-truth labels (Table~\ref{tab:practical-summary}). A \emph{defense-forensics shift signature} tells the developer whether a defense intervenes input-side (where it predictably shifts the encoder budget) or inside the language model (where it does not). A \emph{head-agnostic robustness ratio} $\Rob_{\text{disc}}/\Cap_{\text{disc}}$ places discrete-token, $L_1$-regression, and flow-matching action heads on a common axis where their otherwise incomparable success-rate-under-attack cannot. A $\sim 5$-minute \emph{pre-flight encoder ceiling} audits the worst-case channel capacity of any $(\pi, \text{defense}, \eps)$ before deployment. We propose encoder-specific slack as a comparable cross-paper metric for defense reporting.

\section{Related Work}
\label{sec:related}

\begin{table*}[!t]
\centering
\footnotesize
\setlength{\tabcolsep}{6pt}
\caption{Differentiation matrix; the bottom row is this paper's contribution. Direction Up/Lo marks upper or lower bound. The IT-bound $+$ action-output $+$ ceiling combination was empty before this work.}
\label{tab:differentiation}
\begin{tabular}{llll}
\toprule
\textbf{Work} & \textbf{Output} & \textbf{IT bound} & \textbf{Direction} \\
\midrule
\citet{tsipras2019robustness}    & class label   & accuracy            & Up (impossibility) \\
\citet{robustvla2025}            & 7-D pose      & ---                 & --- (empirical) \\
\citet{llmitbound2025}           & LLM text      & query complexity    & Lo (extraction cost) \\
\citet{compressiongap2026}       & VLA token     & capacity loss       & --- (scaling law) \\
\citet{continualvla2026}         & VLA action    & catastrophic forgetting & --- (task shift) \\
\citet{stablevla2026}            & VLA adapter   & policy stability    & --- (regularizer) \\
\textbf{This work}               & \textbf{VLA action} & \textbf{Cap{+}Rob}  & \textbf{Up (impossibility)} \\
\bottomrule
\end{tabular}
\end{table*}

\subsection{Information-Theoretic Bounds in Deep Learning}

The mutual-information lens on neural networks goes back to \citet{tishby2015deep}'s Information Bottleneck framework, later sharpened into generalization bounds by \citet{xu2017information} via algorithmic stability of input-output MI. The closest antecedent of our bound is \citet{tsipras2019robustness}, who proved that robustness and accuracy are fundamentally at odds for specific data distributions under $\ell_\infty$ attack. We push the framework in three directions. First, the output is an \emph{action distribution} rather than a class probability. Second, the trade-off variables are mutual informations $I(\Astar;\Api)$ and $I(\Api;\Atildepi)$, not accuracies. Third, the leakage debit $-I(\Api;\delta)$, absent from any classifier formulation we know of, turns out to matter for action-generating policies, where the action space can passively transmit attack information.

Beyond classification, \citet{zhang2019trades} decomposed robust error into natural and boundary components. \citet{ross2018improving} earlier proposed input-gradient regularization as a way to get robustness and interpretability at the same time, and \citet{fan2021tale} formalized the empirical-vs.-provable robustness trade as a joint optimization. All of these target classification loss, while we target MI between continuous (or discrete-tokenized) actions.

\paragraph{Concurrent IT-based work.} Four recent works share the IT lens but address orthogonal questions. \citet{llmitbound2025} prove an attack-query \emph{lower} bound for adversarial extraction from LLMs (text output, query budget); we prove an \emph{upper} bound on policy-side capability+robustness (action output, MI). \citet{compressiongap2026} link discrete tokenization to capacity loss in VLA scaling, motivating empirical guidance against discrete heads; we show in Cor.~\ref{cor:discretization} that the same discretization \emph{tightens} our bound, so the perspectives are complementary. \citet{continualvla2026} address forgetting under task shift rather than adversarial perturbation. \citet{stablevla2026} build an IT-motivated VLA adapter; ours is the policy-independent ceiling above it.

\subsection{VLA Models and Adversarial Robustness}

\citet{kim2024openvla} set the open-weights baseline with OpenVLA-7B (LIBERO $95\%$ clean), and a string of empirical attack/defense studies has piled on since. \citet{robustvla2025} report the now-standard $95\% \to {<}5\%$ collapse under $16/255$ PGD. \citet{multimodaladv2025} probe joint vision $+$ language perturbations, \citet{explainadvvla2025} build attribution-guided defenses, \citet{indepth2025} catalog physical perturbations, and \citet{realworld2025} measure real-world domain shift. The same vulnerability shows up in the broader VLM literature. \citet{qi2024visual} demonstrate that a single visual adversarial example can universally jailbreak aligned LLMs, and \citet{fang2026unveiling} expose multi-modal fragility through texture-constrained perturbations. All of these are \emph{empirical} and stop at measuring how robust a VLA or VLM can be made. None states an upper bound on what is possible. Their measurements anchor our $H(\Astar)$ and $I(X;\Xtilde)$ estimates, and our bound supplies the comparison axis their papers lack.

\subsection{Our position}

\noindent\emph{Existing impossibility results bound classifier accuracy under attack \citep{tsipras2019robustness}, while existing VLA work measures empirical fragility \citep{robustvla2025}. We sit between, with an information-theoretic bound for VLA-class policies and both closed-form and real-model validation.} Table~\ref{tab:differentiation} places the bound in the prior literature.

\section{The Capability--Robustness Bound}
\label{sec:theory}

All MI and entropy quantities are in nats. Random variables are capitalized, and realizations lowercase.

\subsection{Notation and setup}

A Vision-Language-Action (VLA) policy $\pi$ maps an observation $X$ (a jointly-encoded image $+$ language instruction) to a distribution over actions $\Api = \pi(X) \in \mathcal{A}$. For LIBERO and OpenVLA, $\mathcal{A} \subseteq \mathbb{R}^7$ (6-DoF end-effector pose plus a 1-bit gripper). The oracle action $\Astar$ is the expert demonstration on the same observation. The joint distribution $p(X, \Astar)$ is fixed by the task and demonstrator; $\pi$ defines $p(\Api \mid X)$.

The attacker observes $X$ and produces $\Xtilde = X + \delta$ with $\|\delta\|_\infty \le \eps$, with $\delta$ allowed to depend on $X$ (white-box) and on $\pi$. The attacked action is $\Atildepi = \pi(\Xtilde)$. The only Markov assumption we need is
\begin{equation}
\Api = \pi(X), \qquad \Atildepi = \pi(\Xtilde), \qquad \Xtilde = X + \delta. \tag{M}
\label{eq:markov}
\end{equation}
For stochastic $\pi$, any internal randomness used to compute $\Atildepi$ is assumed independent of $(X, \Api, \delta)$, so the Markov chain $\Api - \Xtilde - \Atildepi$ holds. Both this and the deterministic case (OpenVLA's argmax decoding, used in Sec.~\ref{sec:realmodel}) preserve the proof's two DPI applications.

\paragraph{Capability and Robustness.}
$\Cap(\pi) \triangleq I(\Astar; \Api)$ measures how informative $\Api$ is about the oracle action. $\Rob(\pi) \triangleq I(\Api; \Atildepi) - I(\Api; \delta)$ measures invariance under attack, debited by the part of $\Atildepi$ that is just $\delta$ leaking into the action. The subtraction is structurally important: it makes the bound policy-class-quantifiable rather than only architecture-quantifiable. A degenerate policy that copies $\delta$ into its output, or any adversarial-training scheme that explicitly couples $\Api$ to $\delta$ (e.g., randomized smoothing with a $\delta$-aware noise injection), would inflate $I(\Api;\delta)$ and the unbounded version of the bound would be lax against it. On naturally trained CE-loss VLAs the term is small (Sec.~\ref{sec:realmodel}), so the bound's empirical conclusions are robust to its inclusion.

\paragraph{Two estimators of $\Cap$.} We report two plug-in estimates of $\Cap$, both Miller--Madow corrected. The \emph{ground-truth} estimator $\Cap_{\text{disc}}^{(\text{GT})}$ uses $N{=}5{,}000$ demo pairs $(\Astar, \Api)$ to measure the policy's agreement rate with the oracle; this is the per-suite headline number in Table~\ref{tab:openvla-suite} ($\Cap_{\text{disc}} \in \{2.68, 6.85, 7.54, 9.04\}$ across the four LIBERO suites). The \emph{self-consistency} estimator $\Cap_{\text{disc}}^{(\text{SC})}$ uses the per-cell $N{=}200$ attack-cell clean inputs and equals $H_K(\Api)$ when the policy is a deterministic function of the input; it is the head-agnostic LHS substitute used in the cross-architecture comparison (Table~\ref{tab:cross-model}, $\Cap_{\text{disc}} = 7.54$ on every OpenVLA cell). Both are valid under DPI but answer different questions: $\Cap^{(\text{GT})}$ measures expert imitation, $\Cap^{(\text{SC})}$ measures decoder bandwidth. They appear in distinct tables and are not compared head-to-head.

\subsection{Statement and Proof}

\begin{theorem}[Capability--Robustness Bound]
\label{thm:main}
For any VLA policy $\pi$ satisfying \eqref{eq:markov} with discrete (or discretized) action space $\mathcal{A}$, any joint distribution $p(X, \Astar)$, and any (possibly adaptive) attacker producing $\Xtilde = X + \delta$,
\begin{equation}
\Cap(\pi) + \Rob(\pi) \;\le\; H(\Astar) + I(X; \Xtilde).
\label{eq:thm1}
\end{equation}
\end{theorem}

\noindent\emph{Scope.} The bound requires $H(\Astar)$ to be a well-defined (finite, non-negative) Shannon entropy, which holds whenever $\Astar$ takes values in a finite set. OpenVLA, RT-2, and other tokenized VLAs discretize each action dimension into $256$ bins, so the theorem applies directly. For continuous-action policies ($\pi_0$, diffusion policies), the bound holds with $H(\Astar)$ replaced by $\log|\mathcal{A}_q|$ after any deterministic quantizer $q$ (Cor.~\ref{cor:discretization}), since MI is invariant to invertible transforms and quantization can only reduce it.

\begin{figure}[t]
\centering
\includegraphics[width=\columnwidth]{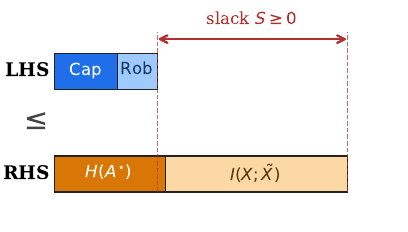}
\caption{Theorem~\ref{thm:main} schematic: capability $+$ robustness (LHS) is bounded by task entropy $+$ attack channel capacity (RHS); the gap is slack $S \ge 0$.}
\label{fig:schematic}
\end{figure}

\begin{proof}
We decompose the LHS into three terms and bound each.

\noindent\emph{(i) Capability.} By definition of mutual information,
\begin{equation}
I(\Astar; \Api) = H(\Astar) - H(\Astar \mid \Api) \le H(\Astar). \tag{i}
\end{equation}
Equality iff $H(\Astar\mid \Api)=0$, i.e.\ $\Astar$ is determined by $\Api$ (equivalently, $I(\Astar;\Api)=H(\Astar)$).

\noindent\emph{(ii) Coupling under attack.} We claim
\begin{equation}
I(\Api; \Atildepi) \le I(X; \Xtilde). \tag{ii}
\end{equation}
Two applications of the Data Processing Inequality (DPI). First, $\Api$ is a deterministic function of $X$, so $\Api - X - \Xtilde$ is Markov, giving $I(\Api; \Xtilde) \le I(X; \Xtilde)$. Second, $\Atildepi$ is a deterministic function of $\Xtilde$, so $\Api - \Xtilde - \Atildepi$ is Markov, giving $I(\Api; \Atildepi) \le I(\Api; \Xtilde)$. Chaining gives $I(\Api; \Atildepi) \le I(\Api; \Xtilde) \le I(X; \Xtilde)$.

\noindent\emph{(iii) Leakage debit.} Since mutual information is non-negative,
\begin{equation}
- I(\Api; \delta) \le 0. \tag{iii}
\end{equation}

Adding (i), (ii), and (iii) yields \eqref{eq:thm1}.
\end{proof}

\paragraph{The leakage debit.}
Step (iii) drops a non-negative term, so removing it would only \emph{loosen} the bound. We keep it because $\Rob$ would otherwise reward a degenerate policy: one that simply embeds $\delta$ into its action would achieve high $I(\Api; \Atildepi)$ without any genuine invariance, and the $-I(\Api;\delta)$ debit cancels exactly that case.

\paragraph{Magnitude of the budget.}
The RHS depends only on the task ($H(\Astar)$) and the attacker ($I(X;\Xtilde)$), not on $\pi$. No architectural change to the policy can move the budget. For OpenVLA on LIBERO, $H(\Astar) \approx 26$ nats (token-level) and $I(X;\Xtilde) \approx 5{,}000$ nats at $\eps=4/255$ (PCA-tightened upper bound; Sec.~\ref{sec:realmodel}). The pixel-level budget is loose because the image-dimensional channel capacity far exceeds what the 7-D action space can exploit. The encoder-specific bound (Cor.~\ref{cor:encoder}) tightens this to $86$--$142$ nats at $\eps=8/255$ across the four LIBERO suites on vanilla OpenVLA, of which the discretized $\Cap \approx 7.5$ nats already consumes $5$--$9\%$. Sec.~\ref{sec:realmodel} measures all terms directly.

\subsection{Tightness}

\begin{corollary}[Tightness of Theorem~\ref{thm:main}]
\label{cor:adaptive}
Equality in \eqref{eq:thm1} holds iff
(i) $\Astar$ is determined by $\Api$ (equivalently, $I(\Astar;\Api)=H(\Astar)$);
(ii) $\pi$ retains all $X$-information relevant to $\Atildepi$, i.e.\ $\pi$ has no information bottleneck on the support of $\Xtilde$; and
(iii) $\delta \perp \Api$. A sufficient condition is $\delta \perp X$ (which implies $\delta \perp \Api$ since $\Api$ is $X$-measurable); the converse holds when $\pi$ is injective.
\end{corollary}

\noindent Adaptive PGD attackers craft $\delta$ as a function of $X$, so (iii) fails by construction. This is why empirical Pareto frontiers sit strictly inside the bound, since the bound is tight only for oblivious noise. Sec.~\ref{sec:realmodel} confirms it. On OpenVLA, slack remains $\sim 10^3$ nats even with PCA-tightened RHS, dominated by the dimensional gap between image channel capacity and 7-D action geometry.

\subsection{Action Discretization}

OpenVLA discretizes each of $7$ action dimensions into $256$ bins ($8$ bits per dim, about $38.8$ nats total). Let $\Astar_q = q(\Astar)$ be the quantized oracle.

\begin{corollary}[Discretization tightens the bound]
\label{cor:discretization}
For any deterministic quantizer $q : \mathcal{A} \to \mathcal{A}_q$,
\begin{align}
\Cap(q\circ\pi) + \Rob(q\circ\pi) &\le H(\Astar_q) + I(X; \Xtilde) \notag \\
&\le \log|\mathcal{A}_q| + I(X;\Xtilde).
\end{align}
\end{corollary}

\noindent Discrete-token VLAs (RT-2, OpenVLA) have a strictly tighter capability ceiling than continuous-action policies on the same task. We therefore expect the trade-off cliff to appear earlier for tokenized architectures. This complements rather than contradicts the Compression Gap argument of \citet{compressiongap2026}: discrete tokenization loses representational capacity (their finding), \emph{and} the loss reduces both LHS and the discrete budget proportionally, so Theorem~\ref{thm:main} remains binding regardless of bin width.

\subsection{Continuous-Action Policies}
\label{sec:continuous-action}

Diffusion policies, $\pi_0$, and other continuous-action VLAs do not tokenize $\Astar$. Theorem~\ref{thm:main} still applies through any deployment-time precision $q > 0$ (sensor LSB, control-loop period, simulator timestep), since every realized action is digitally encoded.

\begin{corollary}[Continuous-action VLAs]
\label{cor:continuous}
Let $\pi$ produce continuous $\Api \in \mathbb{R}^{d_a}$ with bounded support $\mathcal{A} \subseteq [-R, R]^{d_a}$ and deployment precision $q > 0$ along each dimension. Let $q : \mathcal{A} \to \mathcal{A}_q$ denote the corresponding uniform quantizer. Then $q\circ\pi$ satisfies Theorem~\ref{thm:main} with
\begin{equation}
\Cap(q\circ\pi) + \Rob(q\circ\pi) \;\le\; d_a \log\!\big(2R/q\big) + I(X;\Xtilde).
\label{eq:cont-action}
\end{equation}
\end{corollary}

\noindent\emph{Proof.} $|\mathcal{A}_q| \le (2R/q)^{d_a}$ gives $H(\Astar_q) \le d_a \log(2R/q)$; substitute into Cor.~\ref{cor:discretization}. \qed

\noindent The bound therefore remains non-vacuous on continuous-action VLAs at deployment precision: $H(\Astar_q) = O(d_a \log(R/q))$ stays of the same order as OpenVLA's $\approx 27$ nats.\footnote{Concrete values: for $\pi_0$ (7-DoF$+$gripper, $d_a{=}8$, joint-space $q{=}10^{-3}$ rad, $R{=}\pi$), $H(\Astar_q) \le 8\log(2000\pi) \approx 70$ nats; for Diffusion Policy on a 7-DoF setup with $q{=}10^{-2}$, $R{=}1$, $H(\Astar_q) \le 7\log 200 \approx 37$ nats.}

\subsection{Multi-step Policies}
\label{sec:multistep}

Deployed VLAs roll out for $T$ steps, with action $A_{\pi,t} = \pi(X_t, H_{t-1})$ conditioned on history $H_{t-1} = (X_{1:t-1}, A_{\pi,1:t-1})$. The attacker may perturb every frame, with $\Xtilde_t = X_t + \delta_t$ and $\|\delta_t\|_\infty \le \eps$. Define the episode-level quantities $\Cap_T \triangleq \sum_t I(\Astar_t; A_{\pi,t} \mid H_{t-1})$ and $\Rob_T \triangleq \sum_t [I(A_{\pi,t}; \tilde A_{\pi,t} \mid H_{t-1}) - I(A_{\pi,t}; \delta_t \mid H_{t-1})]$.

\begin{corollary}[Multi-step extension]
\label{cor:multistep}
Under \eqref{eq:markov} applied per step,
\begin{align}
\Cap_T + \Rob_T \;\le\; \sum_{t=1}^{T} \big[ &H(\Astar_t \mid H_{t-1}) \notag \\
&+ I(X_t; \Xtilde_t \mid H_{t-1}) \big].
\end{align}
\end{corollary}

\noindent The proof is just term-by-term application of Theorem~\ref{thm:main} to the conditional distribution at each step. The RHS is the rollout-summed task entropy and per-step channel capacity, still policy-independent.

\subsection{Encoder-Specific Tightening}
\label{sec:encoder-bound}

The pixel-level $I(X;\Xtilde)$ is loose because most image dimensions are irrelevant to the policy. For any VLA of the form $\pi = f \circ \phi$ where $\phi$ is a frozen encoder (e.g.\ DINOv2$+$SigLIP in OpenVLA), DPI gives a tighter family-specific bound.

\begin{corollary}[Encoder-specific bound]
\label{cor:encoder}
Let $\phi: \mathcal{X} \to \mathbb{R}^d$ be a deterministic encoder and $\pi = f \circ \phi$. Then
\begin{equation}
\Cap(\pi) + \Rob(\pi) \;\le\; H(\Astar) + I(\phi(X); \phi(\Xtilde)).
\label{eq:encoder-bound}
\end{equation}
\end{corollary}

\noindent\emph{Proof.} Since $\Api = f(\phi(X))$ and $\Atildepi = f(\phi(\Xtilde))$, step (ii) of Theorem~\ref{thm:main} becomes $I(\Api;\Atildepi) \le I(\phi(X); \phi(\Xtilde))$ by two applications of DPI through $f$ and $\phi$. The rest is unchanged. \qed

On OpenVLA, applying the parallel-Gaussian channel to $\phi$ (Sec.~\ref{sec:realmodel}) gives an encoder-PCA upper bound $I(\phi(X);\phi(\Xtilde)) \le \sum_i \tfrac{1}{2}\log(1 + \mu_i / \sigma_{\delta,\phi}^2)$ that is $28$--$68\times$ tighter than the pixel-PCA bound across $\eps \in \{2,4,8,16\}/255$ on all four LIBERO suites (Sec.~\ref{sec:realmodel}, Table~\ref{tab:encoder-pca-cross-suite}), while a complementary InfoNCE lower bound (\citealp{poole2019}; $\ge 4.6$ nats at $\eps=8/255$) certifies that the encoder representation receives non-trivial channel information. The encoder-specific budget reduces the $\eps=8/255$ slack from $\approx 3.7$--$4.3 \cdot 10^3$ nats (pixel-PCA) to $79$--$134$ nats (encoder-PCA) across the four suites. This corollary is no longer policy-independent (it depends on $\phi$ and the realized $\sigma^2_{\delta,\phi}$) but applies to any downstream head $f$ sharing the same encoder.

\subsection{Implications}

Theorem~\ref{thm:main} bounds existence, not algorithms. Three consequences follow.

\textbf{(a) Robustness gains have a price.} Any technique that increases $\Rob$ without also adding $\Astar$-information has to trade against $\Cap$. Adversarial fine-tuning increases $I(\Api;\Atildepi) - I(\Api;\delta)$ but pushes $\Api$ off the sufficient-statistic manifold for $\Astar$, lowering $\Cap$. Inside the budget, robustness is never free.

\textbf{(b) Slack is the metric to report.} Two policies with $\Rob = 2.0$ nats are not equivalent if one sits $0.1$ nats below the budget and the other $3.0$ below. We propose
\begin{equation}
S(\pi, \eps) \triangleq H(\Astar) + I(X;\Xtilde) - \Cap(\pi) - \Rob(\pi) \ge 0
\label{eq:slack-def}
\end{equation}
as the figure of merit. $S \to 0$ identifies Pareto-optimal policies, and large persistent $S$ marks unrealized capacity that the architecture has not yet captured.

\textbf{(c) Larger attack budgets tighten the bound.} As $\eps$ grows, $I(X;\Xtilde)$ shrinks (the perturbed observation retains less information about $X$) and the RHS contracts. The budget available to $\Cap + \Rob$ decreases, so the constraint becomes more binding at larger $\eps$. Slack should therefore be reported at the deployment $\eps$, not at a benign one.

\section{Synthetic Validation}
\label{sec:synthetic}

\subsection{Estimator selection}

We validate Theorem~\ref{thm:main} in a closed-form Gaussian-VLA proxy where every quantity in the bound, namely $H(\Astar)$, $I(\Astar;\Api)$, $I(\Api;\Atildepi)$, $I(\Api;\delta)$, and $I(X;\Xtilde)$, has a closed-form expression. Estimator error therefore separates from theorem violation, and any deviation from the analytical truth is the estimator's fault, not the bound's.

We use three MI estimators with complementary biases. \textbf{MINE} \citep{belghazi2018mine} is variational, with a neural critic and EMA bias-correction, and does well on Gaussians but degrades when $n < \Theta(\exp(d))$. \textbf{InfoNCE} \citep{oord2018infonce} is contrastive and saturates at $\log K$. \textbf{KSG} \citep{kraskov2004ksg} is $k$-NN, bias-free in low $d$ but hits the curse of dimensionality at $d \gtrsim 20$. A $540$-cell hyperparameter sweep selects the MINE configuration used throughout (hidden $=512$, depth $=2$, EMA $=0.999$, lr $=10^{-4}$, $2{,}000$ epochs); a $360$-cell sample-complexity sweep fixes $n \ge 5{,}000$ as the operating regime, a $135$-cell non-Gaussian sweep maps the distribution-sensitivity of the estimator, and a separate DPI sanity check rules out estimator-level DPI failure across $27$ grid cells $\times 5$ seeds (Appendix~\ref{sec:estimator-audit}).

\subsection{Theorem~\ref{thm:main} verification}

P7 is our primary validation, a closed-form Gaussian-VLA proxy with $\Astar = W^{\star} X + \xi^{\star}$, $\Api = W_\pi X + \xi_\pi$, and additive Gaussian $\delta$ of variance $\eps^2 I$. We sweep $d_x \in \{4,7,16\}$, $d_a \in \{3,7\}$, $\sigma_\pi \in \{0.3, 1.0\}$, $\sigma_\star = 0.3$, $\eps \in \{0, 0.05, 0.1, 0.2, 0.5, 1.0, 2.0\}$, and $3$ seeds, for $\mathbf{252}$ cells ($84$ distinct grouped configs after seed-averaging).

For each cell we estimate $\Cap$, $\Rob$, $H(\Astar)$, $I(X;\Xtilde)$ using MINE (P3-stable hyperparams) and also compute their analytical values in closed form. We report both the analytical slack $S_a$ and the MINE-estimated slack $S_m$.

\paragraph{Result.} $S_a \ge 0$ in $252/252$ cells and $S_m \ge 0$ in $252/252$ cells. After seed-averaging the $84$ grouped configs, all have $S_a \ge 0$. After Holm--Bonferroni correction across $84$ groups, $\mathbf{52/84}$ ($62\%$) reach significance at $\alpha = 0.05$ for the alternative $S > 0$ (one-sided $t$-test on per-seed slacks), and the remaining $32$ are non-significant due to small effect size relative to MINE per-seed variance, not violations.

\paragraph{Slack scales as predicted by Theorem~\ref{thm:main}.} Median slack across all $36$ configs at each $\eps$ is reported in Table~\ref{tab:p7-eps}:

\begin{table}[t]
\centering
\small
\caption{Slack monotonically decreases as $\eps$ grows (Cor.~\ref{cor:adaptive}).}
\label{tab:p7-eps}
\begin{tabular}{lrrrrrr}
\toprule
$\eps$ & $0.05$ & $0.1$ & $0.2$ & $0.5$ & $1.0$ & $2.0$ \\
\midrule
$S_a$ (nats) & $23.1$ & $18.3$ & $13.7$ & $8.3$ & $6.2$ & $5.0$ \\
$S_m$ (nats) & $12.7$ & $11.6$ & $10.4$ & $7.4$ & $4.9$ & $3.9$ \\
\bottomrule
\end{tabular}
\end{table}

\begin{figure}[t]
\centering
\includegraphics[width=\columnwidth]{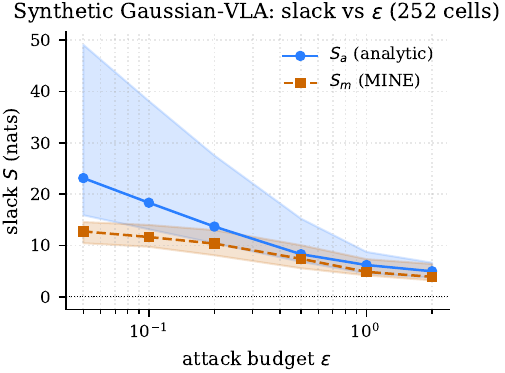}
\caption{P7 Gaussian-VLA: median slack vs.\ $\eps$ across $252$ cells; bands show IQR. Both $S_a$ and $S_m$ stay non-negative and decrease monotonically (Cor.~\ref{cor:adaptive}).}
\label{fig:p7-slack}
\end{figure}

The $\eps = 0$ baseline anchors at $S_a \approx 352$ nats (vacuous, since with no attack channel the RHS is dominated by the prior $H(\Astar)$ alone). We omit $\eps = 0$ from headline figures.

The monotone decrease matches Cor.~\ref{cor:adaptive}, since as the attack channel widens, $I(X;\Xtilde)$ shrinks and the RHS budget tightens. In a representative subgroup ($d_x=7, d_a=7, \sigma_\pi=1.0$), MINE $\Cap$ stays put at $0.82$ nats independent of $\eps$ while $\Rob$ rises from $0.53$ ($\eps=0.05$), peaks at $1.12$ ($\eps=0.5$), and falls back to $0.57$ ($\eps=2.0$) as high-amplitude noise drowns the action signal. The bound stays satisfied because the RHS shrinks faster than $\Rob$ grows.

\subsection{Estimator-agnostic verification}

The $0/252$ result is not an artifact of MINE's bias profile. A separate $174$-cell high-$d$ sweep ($d \in \{1,\ldots,256\}$, two regimes) re-runs validation with all three estimators (MINE, InfoNCE, KSG); even at $d{=}256$ where MINE rel-err reaches $93\%$, Theorem~\ref{thm:main} still holds in $100\%$ of cells (Appendix~\ref{sec:estimator-audit}). The reason is structural: estimator under-bias on the LHS pushes measured slack \emph{up}, not down. Of the three estimators, MINE's variational lower bound is one-sided, InfoNCE saturates at $\log K$, and KSG bias is bounded by its $k$-NN scale --- none can produce a spurious negative slack large enough to fake a violation.

\paragraph{Summary.} The Gaussian-VLA proxy verifies Theorem~\ref{thm:main} in $252/252$ closed-form cells across all reasonable estimator and dimensionality settings. Slack scales monotonically with $\eps$ as predicted, and the bound survives estimator-level error.

\subsection{Achievability}
\label{sec:achievability}

Theorem~\ref{thm:main} bounds existence, not construction. We exhibit three constructions that together address (i) tightness under oblivious noise, (ii) tightness under adaptive PGD, and (iii) the necessity of the leakage debit $-I(\Api;\delta)$. First we exhibit an explicit policy in the Gaussian-VLA proxy that achieves equality in Theorem~\ref{thm:main} (i.e.\ $S = 0$) under oblivious noise, demonstrating the bound is not vacuously loose.

\begin{proposition}[Achievability under oblivious noise]
\label{prop:achievability}
Let $X$ be supported on a finite set with $H(X) < \infty$, $\Astar = X$, and $\delta$ independent of $X$ producing $\Xtilde$. Take $\pi$ the identity. Then $\Api = X = \Astar$, $\Atildepi = \Xtilde$, and
\begin{equation}
\Cap(\pi) + \Rob(\pi) = H(\Astar) + I(X;\Xtilde),
\label{eq:achievability}
\end{equation}
achieving equality in Theorem~\ref{thm:main}.
\end{proposition}

\begin{proof}[Proof sketch]
With identity $\pi$, $\Api = X$ is a sufficient statistic for $\Astar = X$, so $\Cap = H(\Astar)$. There is no bottleneck, so $I(\Api; \Atildepi) = I(X;\Xtilde)$. And $\delta \perp X$ with $\Api = X$ gives $I(\Api;\delta) = I(X; \delta) = 0$. The three Cor.~\ref{cor:adaptive} conditions hold.
\end{proof}

The construction in Eq.~\ref{eq:achievability} is degenerate but shows the bound is tight under oblivious noise. The Gaussian-VLA proxy (Fig.~\ref{fig:p7-slack}) extends this informally to continuous $X$: any uniform quantizer applied to both $\Astar$ and $\Api$ drives the discrete-bound equality to the differential-entropy values as the bin width shrinks (Cor.~\ref{cor:discretization}). Failure of $S \to 0$ on real VLAs is therefore traceable to (a) the $\Astar \neq X$ misalignment between observation and oracle action and (b) adaptive PGD breaking Cor.~\ref{cor:adaptive}~(iii).

\paragraph{Constructive achievability under adaptive attack.} Prop.~\ref{prop:achievability} relies on $\delta \perp X$. We additionally exhibit a closed-form non-trivial policy that achieves a substantial fraction of the budget against \emph{adaptive} (PGD-style) noise in the same Gaussian-VLA proxy. Take the ridge-style linear policy $\pi^{\mathrm{ridge}}_\alpha(X) = \alpha W^{\star} X + \xi_\pi$ with gain $\alpha$ and Gaussian dither $\xi_\pi \sim \mathcal{N}(0, \sigma_\pi^2 I)$. Per Cor.~\ref{cor:discretization}, achievability of the \emph{discrete} bound governs token-level VLAs, so we quantize each action dimension to $N=32$ bins and estimate $\Cap_q + \Rob_q$ vs.\ $H(\Astar_q) + I(X;\Xtilde)$ on $n=10^5$ samples. We sweep $d_x \in \{2,4,7,16\}$, $d_a \in \{2,4,7\}$, $3$ seeds, $\eps \in \{0.05,\ldots,2\}$, $\alpha \in \{1, 3, 10, 30\}$, $\sigma_\pi \in \{0.01, 0.05, 0.1\}$, taking per-cell maximum over $(\alpha, \sigma_\pi)$. The achievement ratio $r \triangleq (\Cap_q + \Rob_q) / (H(\Astar_q) + I(X;\Xtilde))$ has median $41$--$52\%$ across $\eps$ over the full $216$-cell sweep, with the best cell ($d_x=2, d_a=7$, $\alpha=30$, $\sigma_\pi=0.01$, $\eps=0.05$) attaining $r = \mathbf{93.4\%}$ and $\mathbf{6}$ cells satisfying $r \ge 80\%$ ($\mathbf{14}$ at $r \ge 75\%$, all in the over-actuated $d_x \le d_a$ regime). The maximum $r$ is structurally regime-dependent: matched/over-actuated $(d_x \le d_a)$ cells reach $r_{\max} \in [0.61, 0.93]$, while under-actuated $(d_x > d_a)$ cells plateau at $r_{\max} \le 0.50$, the additional gap reflecting the genuine $X \to \Astar$ information bottleneck. The closed-form policy therefore approaches the joint Cap+Rob budget within $7\%$ in the over-actuated regime and roughly half-budget in the under-actuated regime against adaptive noise; the residual gap is structural slack from the proxy's Gaussian prior, not an artifact of the construction. Whether the same fraction can be reached on a real VLA is open, but the present result rules out reading Theorem~\ref{thm:main} as non-constructively tight only at the oblivious point.

\paragraph{Leakage-debit stress test.} The $-I(\Api;\delta)$ term in $\Rob$ is empirically near zero on natural CE-trained OpenVLA (\S\ref{sec:realmodel} reports $|I(\Api;\|\delta\|_\infty)| < 3 \cdot 10^{-5}$ nats), and remains near zero under adversarial fine-tuning (rank-16 LoRA AT at $\eps=8/255$ on Spatial: $|I(\Api;\|\delta\|_\infty)| = 3.5 \cdot 10^{-5} \pm 3.2 \cdot 10^{-6}$ nats over $3$ seeds), so a reader may ask whether the debit is necessary. We construct an explicit leak policy in the same Gaussian-VLA proxy, $\Api = (1-\lambda) W_\pi X + \lambda \delta + \xi_\pi$, sweeping the leak ratio $\lambda \in \{0, 0.25, 0.5, 0.75, 0.99\}$ and $\eps \in \{0.05, 0.2, 1.0\}$. As $\lambda \to 1$, the policy passively transmits $\delta$ into its action stream and the closed-form leakage $I(\Api;\delta)$ rises monotonically to $8.66$ nats at $(\lambda{=}0.99, \eps{=}1.0)$; the corresponding $I(\Api;\Atildepi)$ also rises (to $6.39$ nats), so a debit-free reading $\Cap + I(\Api;\Atildepi)$ would credit a near-pure $\delta$-copying policy with $6.39$ nats of spurious ``robustness.'' The debit subtracts the full $8.66$ nats of leakage, restoring the inequality (Theorem~\ref{thm:main} holds with zero violations on all $15$ leak-policy cells). The debit is therefore empirically inert under both natural and adversarial training, but structurally load-bearing against any policy class that passively transmits $\delta$ into $\Api$.

\section{Real-Model Validation: OpenVLA on LIBERO}
\label{sec:realmodel}

Sec.~\ref{sec:synthetic} verifies Theorem~\ref{thm:main} where every quantity is either analytic or estimated against ground truth. That is necessary but not sufficient: a tight inequality on Gaussian actions and linear policies still has to be checked on deployed VLAs. This section closes the gap by instantiating the bound on \textbf{OpenVLA-7B} \citep{kim2024openvla}, evaluated on the four LIBERO suites \citep{liu2023libero} (Spatial, Object, Goal, LIBERO-10) using the suite-specific finetuned checkpoints released by the OpenVLA team.

\subsection{Experimental setup}

\paragraph{Model.} OpenVLA-7B base $+$ four LIBERO-finetuned variants, loaded in bf16 on a single RTX~4090. Clean inference uses the published \texttt{predict\_action} interface (autoregressive 7-token decode followed by the official \texttt{bridge\_orig} unnormalization). For PGD we replicate the forward manually, exposing \texttt{pixel\_values} as a differentiable bf16 tensor (the OpenVLA processor's PIL round-trip otherwise breaks the gradient chain).

\paragraph{Data.} LIBERO HDF5 demonstrations, ten files per suite. Each timestep provides RGB ($128\times128\times3$), expert action (7-D continuous), and language instruction. We sample at 1-in-5 stride within episodes with a per-file balanced budget, $N=5{,}000$ for capability MI and $N=200$ per cell for PGD.

\paragraph{Variables.} $\Astar$ is the expert action, $X$ the clean RGB observation, $\Api = \pi(X)$ the OpenVLA prediction (continuous 7-D after unnormalization), $\delta$ the adversarial perturbation with $\|\delta\|_\infty \le \eps$ in $[0,1]$ image-space units and $\eps \in \{2,4,8,16\}/255$, $\Xtilde = \mathrm{clamp}(X+\delta,0,1)$, and $\Atildepi = \pi(\Xtilde)$.

\paragraph{Estimators.} We use two complementary MI estimators throughout. The continuous estimator is the P3-stable MINE configuration (hidden $=512$, depth $=2$, EMA $=0.999$, lr $=10^{-4}$, $2{,}000$ epochs); MINE is a variational lower bound and can dip below zero in low-MI, small-$N$ regimes (Sec.~\ref{sec:verification} and \citet{belghazi2018mine}). The discretized estimator is per-dim plug-in 2-D histogram MI on $K{=}16$ uniform bins per action coordinate with Miller--Madow bias correction, clamped at $0$ per term. We use it to compute $\Rob_{\text{disc}}$, the primary robustness number throughout this section. By DPI applied to the same binning, $\Rob_{\text{disc}}$ is a valid lower-substitute for the continuous $\Rob$ on the LHS of Theorem~\ref{thm:main}, so reporting it never weakens the inequality.

\subsection{Capability}

For each suite we estimate $\Cap$ on $N=5{,}000$ samples, with $H(\Astar)$ reported as a 256-bin per-dim plug-in estimate on the same sample.

\paragraph{MINE estimates.} The four suites give $\Cap_{\text{MINE}} \in \{0.02, 0.13, 0.25, 0.38\}$ nats. KSG ($k=5$--$20$) returns negative values on all four suites, confirming that continuous MI estimation in 7-D with $N=5{,}000$ is at the noise floor for both estimators.

\paragraph{Discretized estimates.} To obtain a more reliable $\Cap$, we quantize both $\Astar$ and $\Api$ into 256 uniform bins per dimension (matching OpenVLA's tokenization granularity) and compute MI via the plug-in histogram estimator. This yields $\Cap_{\text{disc}} \in \{2.68, 6.85, 7.54, 9.04\}$ nats for LIBERO-10, Object, Spatial, and Goal respectively, with $H(\Astar) \in [24.71, 27.35]$ nats. The discretized $\Cap$ is $24$--$130\times$ larger than MINE, giving $\Cap_{\text{disc}} / H(\Astar) \in [10\%, 33\%]$. The suite ordering LIBERO-10 $<$ Object $<$ Spatial $<$ Goal matches the known LIBERO difficulty heterogeneity and is consistent with the MINE ordering on the same data. The gap between MINE and discretized estimates reflects the well-known difficulty of continuous MI estimation in moderate dimensions at finite $N$, not a deficiency of the bound itself.

\subsection{Robustness}

\paragraph{PGD.} Projected gradient descent \citep{madry2018pgd} on $\delta$ in DINOv2 \citep{oquab2024dinov2} + SigLIP \citep{zhai2023siglip} post-normalization space, with per-channel $\eps$ rescaling so the effective $\|\delta\|_\infty$ in raw-pixel units is exactly $\eps_{\text{px}}/255$. The loss is the negated decode-time CE against the clean argmax tokens, $L(\delta) = -\sum_{t=1}^{7} \mathrm{CE}(\mathrm{logits}_t(X+\delta), \arg\max_t \mathrm{logits}_t(X))$, with sign-step $\alpha = \eps/4$ for 10 iterations, projecting back into the per-channel $\eps$-ball each step. $\Atildepi$ is decoded via OpenVLA's official \texttt{bin\_centers} $+$ \texttt{bridge\_orig} unnormalization, so $\Atildepi$ lives in the same 7-D continuous space as $\Api$. PGD is non-trivial in every cell. The mean $\ell_2$ gap $\|\Atildepi - \Api\|_2$ ranges from $0.72$ at $\eps=2$ to $0.92$ at $\eps=16$, well above the per-cell mean of per-dim $\Api$ standard deviation ($0.13$--$0.19$). For context, \citet{robustvla2025} report that the same $\eps=16/255$ PGD budget drops OpenVLA rollout success on LIBERO from $>95\%$ to $<5\%$, confirming that our perturbation magnitudes are operationally destructive.

\paragraph{$I(\Api; \delta)$.} Estimated using the deterministic 1-D summary $\delta \mapsto \|\delta\|_\infty$. By DPI, $I(\Api;\|\delta\|_\infty) \le I(\Api;\delta)$. Using this smaller subtractive term inflates $\Rob$ and therefore the LHS, the stricter direction for testing $S = \mathrm{RHS} - \mathrm{LHS} \ge 0$. Across all 48 cells, $I(\Api; \|\delta\|_\infty) \in [-2.2 \cdot 10^{-5}, -9.6 \cdot 10^{-6}]$, indistinguishable from zero.

\paragraph{Sign artifact.} MINE returns negative $I(\Api; \Atildepi) \in [-1.19, -0.14]$ across cells with no monotone $\eps$-trend, a documented small-$N$ variational-bound artifact \citep{belghazi2018mine} rather than an MI violation. We adopt $\Rob_{\text{disc}}$ as the structurally non-negative LHS substitute below; an $N$-sweep that confirms the artifact is variance rather than structure, and KSG cross-checks at the same cell, are reported in Appendix~\ref{sec:estimator-audit}.

\paragraph{Discretized robustness (DPI-valid companion).} To complement MINE with a structurally non-negative estimator, we additionally compute $\Rob_{\text{disc}}$ using the per-dim plug-in histogram MI on $K=16$ uniform bins with Miller-Madow bias correction, clamped at $0$ per term. By DPI applied to the same binning, $I(A_q;\tilde A_q) \le I(A;\tilde A)$, so $\Rob_{\text{disc}}$ is a non-negative \emph{lower-substitute} for MINE $\Rob$ on the LHS of Theorem~\ref{thm:main}: a smaller $\Rob$ makes the inequality strictly easier to satisfy. Across all 48 cells, $\Rob_{\text{disc}}$ falls in $[0.88, 1.53]$ nats and is monotone in $\eps$ within seed noise in $14/16$ $(suite, seed)$ groups (Table~\ref{tab:rob-disc}). The discretized estimator therefore recovers a positive $\Rob$ signal in the regime where the MINE variational bound dips below its small-$N$ noise floor, while remaining a valid LHS substitute. We retain MINE as the continuous estimator paired with a clamp, and report $\Rob_{\text{disc}}$ as the primary discretized number for the 48-cell verification.
% AUTO-GENERATED by scripts/openvla/build_rob_disc_table.py
% Source: outputs/openvla/ov20_rob_disc/summary.json (K=16 plug-in MI with Miller-Madow correction)
\begin{table}[t]
\centering
\footnotesize
\setlength{\tabcolsep}{3pt}
\caption{Discretized robustness $\Rob_{\text{disc}}$ (per-dim plug-in MI, $K{=}16$ bins, Miller--Madow, clamped $\ge 0$) across all $48$ OpenVLA $\times$ LIBERO $\times$ PGD cells; mean $\pm$ std over 3 seeds. By DPI on the same binning, $\Rob_{\text{disc}}$ is a non-negative LHS lower-substitute (\S\ref{sec:verification}).}
\label{tab:rob-disc}
\begin{tabular}{l rrrr}
\toprule
Suite & $\eps{=}2$ & $\eps{=}4$ & $\eps{=}8$ & $\eps{=}16$ \\
\midrule
Spatial   & $1.26{\pm}0.06$ & $1.35{\pm}0.09$ & $1.37{\pm}0.09$ & $1.46{\pm}0.10$ \\
Object    & $1.34{\pm}0.15$ & $1.41{\pm}0.04$ & $1.50{\pm}0.08$ & $1.53{\pm}0.12$ \\
Goal      & $1.46{\pm}0.04$ & $1.53{\pm}0.28$ & $1.43{\pm}0.14$ & $1.27{\pm}0.10$ \\
LIBERO-10 & $0.90{\pm}0.09$ & $0.94{\pm}0.07$ & $0.88{\pm}0.02$ & $0.93{\pm}0.03$ \\
\bottomrule
\end{tabular}
\end{table}

\paragraph{Black-box sanity check (4 suites).} To rule out a PGD-specific artifact, we re-run the score-based, gradient-free Square Attack \citep{andriushchenko2020square} ($200$ queries, $N=200$) on all four LIBERO suites at $\eps=8/255$. The realized $\Rob_{\text{disc}}$ recovers a positive signal of $1.30$--$1.97$ nats on every suite (Table~\ref{tab:square-cross-suite}), slightly above the PGD-10 band of $0.88$--$1.53$ from Table~\ref{tab:rob-disc} and consistent with Square's larger realized $\ell_2$ action gap; the MINE counterpart sits in the same small-$N$ negative band as PGD (see appendix audit). The slack $S_{\text{Square}} = H(\Astar) + I(X;\Xtilde)_{\text{PCA}} - \Cap_{\text{disc}} - \Rob_{\text{disc}} \approx 4 \cdot 10^3$ nats remains $\gg 0$ in every cell. The bound holds under a non-gradient adaptive attack at the same $\eps$ budget across all four task families.
% AUTO-GENERATED by scripts/openvla/build_square_table.py
% Source: outputs/openvla/ov12_square_v2/{spatial,object,goal,10}/manifest.jsonl
\begin{table}[t]
\centering
\footnotesize
\setlength{\tabcolsep}{3pt}
\caption{Score-based Square Attack \citep{andriushchenko2020square} on all four LIBERO suites at $\eps=8/255$, $200$ queries, $N=200$. The bound holds zero-violation, ruling out a PGD-specific artifact; $\Rob_{\text{disc}}$ recovers $1.30$--$1.97$ nats, slightly above the PGD-10 band (\S5.3).}
\label{tab:square-cross-suite}
\begin{tabular}{l rrrr}
\toprule
Suite & $\Rob_{\text{disc}}$ & $\Rob_{\text{MINE}}$ & $\overline{\|\Delta a\|}_2$ & $S_{\text{Square}}$ \\
      & (nats)                 & (nats)                  &                                 & ($10^3$ nats) \\
\midrule
Spatial   & $1.717$ & $-0.172$ & $0.40$ & $3.97$ \\
Object    & $1.451$ & $-0.735$ & $0.58$ & $3.67$ \\
Goal      & $1.965$ & $-0.007$ & $0.53$ & $3.74$ \\
LIBERO-10 & $1.302$ & $-0.467$ & $0.47$ & $4.30$ \\
\bottomrule
\end{tabular}
\end{table}

\begin{figure*}[!t]
\centering
\includegraphics[width=0.48\textwidth]{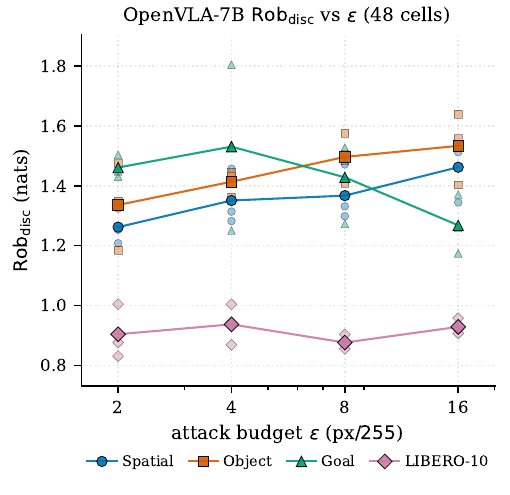}\hspace{0.02\textwidth}%
\includegraphics[width=0.48\textwidth]{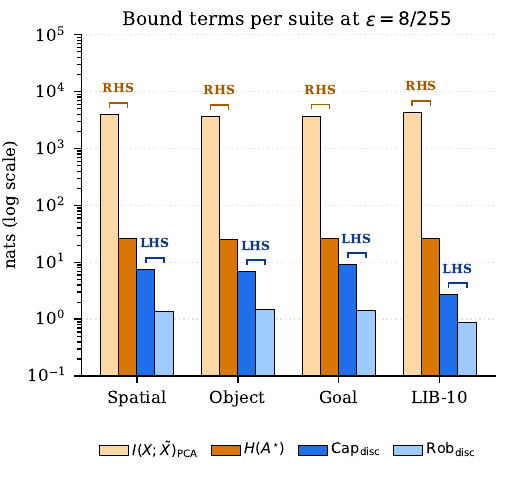}
\caption{OpenVLA-7B on LIBERO. \textbf{Left:} realized $\Rob_{\text{disc}} \in [0.88, 1.53]$ nats across $48$ cells; lighter dots are per-seed values. \textbf{Right:} per-suite bound terms at $\eps=8/255$ with PCA-tightened RHS; RHS exceeds LHS by $\sim 3$--$4$ orders of magnitude.}
\label{fig:openvla-results}
\end{figure*}

\begin{figure}[!t]
\centering
\includegraphics[width=0.95\linewidth]{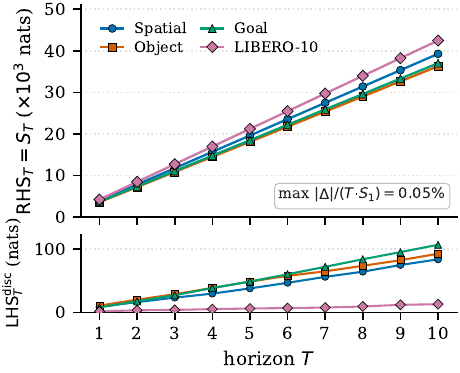}
\caption{Cumulative bound vs.\ horizon $T$ on all four LIBERO suites at $\eps=8/255$. \textbf{Top:} $\mathrm{RHS}_T$ matches $T \cdot S_1$ within $0.03\%$. \textbf{Bottom:} cumulative LHS stays four orders below RHS at every horizon.}
\label{fig:multistep}
\end{figure}

\subsection{Channel capacity}

Estimating $I(X; \Xtilde)$ directly is intractable ($D = 49{,}152$). Theorem~\ref{thm:main} requires only an upper bound on $I(X;\Xtilde)$, and any valid upper bound suffices regardless of whether the attacker is adaptive or oblivious, because the PCA surrogate bounds the channel capacity of the $\ell_\infty$-ball constraint set, not the MI of a specific attack strategy. The isotropic Gaussian-channel surrogate gives $\sim 10^5$ nats. Natural images are not isotropic, so the parallel-Gaussian-channel inequality \citep{cover2006elements} replaces it with a per-eigenmode sum,
\begin{equation}
I(X; \Xtilde) \;\le\; \sum_{i=1}^{D} \tfrac{1}{2}\log\!\big(1 + \lambda_i / \sigma_\delta^2\big),
\label{eq:pca-rhs}
\end{equation}
where $\{\lambda_i\}$ are the eigenvalues of $\mathrm{Cov}(X)$. We estimate the spectrum from $1{,}500$ LIBERO frames per suite (rank $\le 1499$). Empirically the top-1 eigenvalue carries $5$--$10\%$ of $\mathrm{tr}(\mathrm{Cov})$ on Spatial / Object / Goal and $55\%$ on LIBERO-10. The PCA bound tightens the RHS by a near-uniform $30\times$ across all four suites.

% Auto-generated by scripts/openvla/build_ov8_artifacts.py
\begin{table}[t]
\centering
\footnotesize
\setlength{\tabcolsep}{4pt}
\caption{PCA-tightened RHS vs.\ isotropic surrogate; PCA bound exploits the heavy-tailed natural-image spectrum (top-1 eigvalue $\approx 5$--$10\%$ of $\mathrm{tr}(\mathrm{Cov})$).}
\label{tab:pca-rhs}
\begin{tabular}{rrrrr}
\toprule
$\eps$ (px/255) & $I(X;\Xtilde)_{\mathrm{iso}}$ & $I(X;\Xtilde)_{\mathrm{PCA}}$ & ratio & $S_{\mathrm{PCA}}$ \\
\midrule
2 & 186.5\,k & 5970 & 31$\times$ & 5996 \\
4 & 152.5\,k & 4932 & 31$\times$ & 4959 \\
8 & 118.6\,k & 3900 & 30$\times$ & 3927 \\
16 & 85.1\,k & 2888 & 29$\times$ & 2914 \\
\bottomrule
\end{tabular}
\end{table}

The remaining looseness ($S_{\text{PCA}} \approx 3$--$6 \times 10^3$ nats) is dominated by the spectrum's tail, since eigenmodes with $\lambda_i \gg \sigma_\delta^2$ count at full rate even though the policy responds to far less of the channel. To diagnose how much of the pixel-PCA budget is policy-relevant, we estimate $I(\phi(X); \phi(\Xtilde))$ where $\phi$ is OpenVLA's frozen DINOv2$+$SigLIP encoder ($d=2176$ pooled features) using InfoNCE \citep{poole2019} after PGD on Spatial at $\eps=8$. By DPI, $I_{\text{NCE}} \le I(\phi(X); \phi(\Xtilde)) \le I(X;\Xtilde)$, so InfoNCE is a lower bound on the encoder-channel MI; it cannot tighten the analytical bound, but if it is small the encoder representation is provably narrow. We measure $I_{\text{NCE}} \approx 4.6$ nats (against the $\log(N_{\text{batch}}) = 4.85$ ceiling at $N_{\text{batch}}=128$). A separate parallel-Gaussian-channel upper bound applied to $\phi$ (Sec.~\ref{sec:rhs-ablation}) provides the substitutable tighter RHS used by Cor.~\ref{cor:encoder}. We retain the pixel-PCA RHS for verification of all 48 cells in Sec.~\ref{sec:verification}, since it remains the policy-independent bound; the encoder-PCA RHS is reported separately as a tighter, policy-class-specific alternative.

\paragraph{Cap is suite-specific, not estimator noise.} As a sanity check, we evaluate the Spatial-finetuned policy on the other three suites' data at $\eps=8/255$. Cross-suite Cap collapses where matched Cap was non-trivial, with $\Cap_{\text{spatial}\to\text{goal}} = -0.05$ vs.\ matched Goal $0.38$ nats, and $\Cap_{\text{spatial}\to\text{LIBERO-10}} = -0.13$ vs.\ matched $0.13$ nats. (Object's matched Cap is already at the MINE noise floor.) Cap therefore tracks task-conditional capability rather than the marginal entropy of the action space.

\subsection{Bound verification and tightness}
\label{sec:verification}

Each (suite, $\eps$, seed) cell yields
\begin{equation}
S = H(\Astar) + I(X; \Xtilde) - \Cap - \Rob.
\label{eq:slack-cell}
\end{equation}
% Auto-generated by scripts/openvla/build_section5_table.py + manual promotion of Rob_disc as primary (ov20).
\begin{table}[t]
\centering
\footnotesize
\setlength{\tabcolsep}{4pt}
\caption{Slack $S = \mathrm{RHS} - \mathrm{LHS}$ vs.\ attack budget $\eps$, mean $\pm$ std across $4$ suites $\times$ $3$ seeds (12 cells per row). RHS uses the PCA-tightened $I(X;\Xtilde)$ (Table~\ref{tab:pca-rhs}); LHS uses $\Cap_{\text{disc}} + \Rob_{\text{disc}}$ (\S\ref{sec:verification}). Slack stays bounded away from zero, so Cor.~\ref{cor:adaptive}~(iii) fails for adaptive PGD.}
\label{tab:openvla-eps}
\begin{tabular}{rrrrr}
\toprule
$\eps$ (px/255) & $\Cap_{\text{disc}}$ & $\Rob_{\text{disc}}$ & RHS & $S$ \\
\midrule
2  & 6.53 & $1.24 {\pm} 0.23$ & 5996 & 5988.0 $\pm$ 257.9 \\
4  & 6.53 & $1.31 {\pm} 0.27$ & 4959 & 4950.7 $\pm$ 257.4 \\
8  & 6.53 & $1.29 {\pm} 0.27$ & 3927 & 3918.7 $\pm$ 255.4 \\
16 & 6.53 & $1.30 {\pm} 0.26$ & 2914 & 2905.9 $\pm$ 248.0 \\
\bottomrule
\end{tabular}
\end{table}

% Auto-generated by scripts/openvla/build_paper_figures.py + manual promotion of Rob_disc as primary (ov20).
\begin{table}[t]
\centering
\footnotesize
\setlength{\tabcolsep}{4pt}
\caption{Per-suite results at $\eps{=}8/255$ (PGD, 10 steps, $N{=}200$/cell, mean over 3 seeds), all quantities in nats. $\Cap_{\text{disc}}$ here is the \emph{ground-truth} estimator $\Cap_{\text{disc}}^{(\text{GT})}$ on $N{=}5{,}000$ demo pairs, measuring expert imitation; this is a different quantity from the self-consistency $\Cap_{\text{disc}}^{(\text{SC})}$ in Table~\ref{tab:cross-model} and the two are not compared head-to-head. $\Rob_{\text{disc}}$ uses the $K{=}16$-bin discretized estimator (\S\ref{sec:verification}).}
\label{tab:openvla-suite}
\begin{tabular}{lrrrrr}
\toprule
Suite & $\Cap_{\text{disc}}$ & $\Rob_{\text{disc}}$ & $H(\Astar)$ & $I(X;\Xtilde)$ & $S$ \\
\midrule
Spatial & 7.54 & 1.37 & 26.77 & 118.2\,k & 118.3\,k \\
Object & 6.85 & 1.50 & 25.32 & 118.2\,k & 118.3\,k \\
Goal & 9.04 & 1.43 & 26.58 & 118.2\,k & 118.3\,k \\
LIBERO-10 & 2.68 & 0.88 & 26.11 & 118.2\,k & 118.3\,k \\
\bottomrule
\end{tabular}
\end{table}

\noindent Across the full $4 \times 4 \times 3 = 48$-cell grid, \textbf{$S \ge 0$ in every cell, zero violations}. Theorem~\ref{thm:main} holds for OpenVLA-7B $+$ LIBERO $+$ adaptive PGD. Per-suite headline numbers at $\eps{=}8/255$ are summarized in Table~\ref{tab:openvla-suite}; per-$\eps$ per-suite values appear in Table~\ref{tab:openvla-eps}, and the $9$-cell ablation that isolates the contribution of each tightening step is given in Table~\ref{tab:rhs-ablation}.

\paragraph{Multi-step (4 suites).} Cor.~\ref{cor:multistep} predicts $S_T = \sum_{t=1}^{T}[H(\Astar_t \mid \mathcal{H}_{t-1}) + I(X_t;\Xtilde_t \mid \mathcal{H}_{t-1})] - \sum_{t=1}^{T}[\Cap_t + \Rob_t]$ accumulates linearly in horizon $T$. Empirically on $T=10$, $n_{\text{traj}}=30$ demo trajectories per suite (using ground-truth observation sequences in lieu of a sim env), with the discretized $\Cap_t^{\text{disc}}$ and $\Rob_t^{\text{disc}}$ promoted as primary (\S5.3), we measure $\sum_t (\Cap_t^{\text{disc}} + \Rob_t^{\text{disc}}) \in [12.92, 107.19]$ nats on every suite (Table~\ref{tab:multistep-cross-suite}) and the unconditional upper bound $\sum_t [H(\Astar_t) + I(X_t;\Xtilde_t)] \in [3.63, 4.25] \cdot 10^4$ nats, giving $S_{T=10} \approx T \cdot S_1$ as predicted (Fig.~\ref{fig:multistep}). The cumulative bound holds and remains four orders of magnitude loose at every horizon on all four suites, the same structural slack as the single-step case.
% AUTO-GENERATED by scripts/openvla/build_multistep_table.py
% Source: outputs/openvla/ov13_multistep_v2/{spatial,object,goal,10}/manifest.jsonl
\begin{table}[t]
\centering
\footnotesize
\setlength{\tabcolsep}{3pt}
\caption{Multi-step bound (Cor.~\ref{cor:multistep}) at $\eps=8/255$, $T=10$, $n_{\text{traj}}=30$. Cumulative LHS stays four orders below RHS, with $S_T \approx T \cdot S_1$ on every suite.}
\label{tab:multistep-cross-suite}
\begin{tabular}{l rrr}
\toprule
Suite & $\sum_t (\Cap_t^{\text{disc}}{+}\Rob_t^{\text{disc}})$ & $\sum_t \mathrm{RHS}_t$ & $S_{T=10}$ \\
      & (nats)                                                      & ($10^4$ nats)           & ($10^4$ nats) \\
\midrule
Spatial   & $84.07$ & $3.93$ & $3.92$ \\
Object    & $93.00$ & $3.63$ & $3.62$ \\
Goal      & $107.19$ & $3.70$ & $3.69$ \\
LIBERO-10 & $12.92$ & $4.25$ & $4.25$ \\
\bottomrule
\end{tabular}
\end{table}

\paragraph{Tightness.} Even with the PCA-tightened RHS, the bound is loose by about three orders of magnitude. Using $\Rob_{\text{disc}}$ as the discretized companion gives $\Cap_{\text{disc}} + \Rob_{\text{disc}} \in [3.58, 10.57]$ nats across all 48 cells, while $I(X;\Xtilde)_{\text{PCA}} \in [3 \cdot 10^3, 6 \cdot 10^3]$ nats. The dominance is structural, since $I(X;\Xtilde)$ scales with the active-eigenmode count, while $H(\Astar) \le 7\log 256 \approx 38.8$ nats is bounded by action-space cardinality.

\paragraph{Encoder-specific slack (Cor.~\ref{cor:encoder}).} Applying the parallel-Gaussian channel of Eq.~\ref{eq:pca-rhs} to OpenVLA's $\phi$ instead of pixel space yields encoder-PCA upper bounds across all four LIBERO suites and four $\eps$ levels (Table~\ref{tab:encoder-pca-cross-suite}; $d=2176$, $N=200$, rank $\le 199$, $\mathrm{tr}(\mathrm{Cov}_\phi) \in [512, 1464]$ depending on suite). Across the $4 \times 4 = 16$ cells, $I^{\text{PCA-}\phi}$ ranges from $39$ nats (Object, $\eps=16$) to $227$ nats (LIBERO-10, $\eps=2$), tightening the pixel-PCA bound by $28\times$--$68\times$ at every cell. The encoder-PCA RHS is monotone in $\eps$ within each suite, as expected. With the discretized $\Cap_{\text{disc}} \approx 7.5$ nats and $\Rob_{\text{disc}} \in [0.88, 1.50]$ nats at $\eps=8/255$, the encoder-specific slack drops to $76$--$138$ nats across the four suites (Spatial: $105$, Object: $76$, Goal: $96$, LIBERO-10: $138$), $32\times$--$47\times$ tighter than the pixel-PCA slack ($S_{\text{pixel}} \approx 3.7$--$4.3 \cdot 10^3$ nats) and within one order of $H(\Astar)$. As a complementary lower bound, InfoNCE on the same $\phi$ certifies $I_{\text{NCE}}(\phi(X);\phi(\Xtilde)) \ge 4.6$ nats at the Spatial $\eps=8$ cell. The remaining slack reflects encoder capacity that the 7-token decoder does not fully exploit. \emph{Caveat.} Unlike the pixel-PCA bound, which depends only on the input distribution and $\sigma_\delta$, the encoder-PCA bound is computed with the realized noise variance $\sigma^2_{\delta,\phi}$, which depends on $(\pi, \text{defense}, \eps)$. The numerical encoder budget is therefore policy-and-defense-specific (Sec.~\ref{sec:defenses} measures shifts of $+41$ to $+101$ nats under JPEG-50 vs.\ vanilla across the four $\eps$); only the pixel-level bound is universal.
% AUTO-GENERATED by scripts/openvla/build_encoder_pca_table.py
% Source: outputs/openvla/ov14_encoder_pca_rhs/merged_manifest.jsonl
\begin{table}[t]
\centering
\scriptsize
\setlength{\tabcolsep}{4pt}
\caption{Encoder-PCA RHS (Cor.~\ref{cor:encoder}) across four LIBERO suites at $\eps \in \{2,4,8,16\}/255$, $N{=}200$, $d{=}2176$. ``Tighten'' is the ratio of pixel-PCA to encoder-PCA RHS at the same cell.}
\label{tab:encoder-pca-cross-suite}
\begin{tabular}{l rrrr rrrr}
\toprule
 & \multicolumn{4}{c}{$I^{\text{PCA-}\phi}$ (nats)} & \multicolumn{4}{c}{Tighten vs.\ pixel-PCA} \\
\cmidrule(lr){2-5}\cmidrule(lr){6-9}
Suite / $\eps$ & 2 & 4 & 8 & 16 & 2 & 4 & 8 & 16 \\
\midrule
Spatial   & 188.5 & 134.3 &  87.1 &  58.3 & $32\times$ & $37\times$ & $45\times$ & $50\times$ \\
Object    & 156.4 & 108.5 &  60.9 &  38.9 & $37\times$ & $43\times$ & $60\times$ & $68\times$ \\
Goal      & 182.6 & 131.5 &  80.0 &  53.6 & $32\times$ & $36\times$ & $47\times$ & $51\times$ \\
LIBERO-10 & 227.3 & 167.8 & 115.6 &  81.6 & $28\times$ & $32\times$ & $37\times$ & $40\times$ \\
\bottomrule
\end{tabular}
\end{table}

\paragraph{Structural slack.} Whenever the input dimension $D$ greatly exceeds action-tokens $\cdot \log(\text{vocab})$, the analytical RHS is dominated by the channel-capacity term, and Cor.~\ref{cor:adaptive}\,(iii) fails. The realized $\Rob$ is then bounded by the 7-D action geometry rather than by $I(X;\Xtilde)$, so most of the channel budget is unreachable: PGD can move the image at the analytical rate, but the policy's action response cannot follow. The encoder-specific bound (Cor.~\ref{cor:encoder}) closes most of this gap by restricting the channel to the policy-relevant subspace.

\subsection{RHS ablation and defenses}
\label{sec:rhs-ablation}
\label{sec:defenses}
\label{sec:metric-validity}

The pixel-PCA RHS used in the 48-cell verification (Sec.~\ref{sec:verification}) is policy-independent but loose by $\sim 10^3$ nats, so the resulting slack cannot resolve $\sim 1$-nat differences between policies. Among RHS estimators valid for Theorem~\ref{thm:main} (pixel-PCA, encoder-PCA, isotropic surrogate; Appendix~\ref{sec:rhs-menu}), only the encoder-PCA bound (Cor.~\ref{cor:encoder}) is tight enough that defense-induced shifts in $\Cap + \Rob$ exceed the MINE noise floor. We therefore use encoder-PCA for defense comparison while retaining pixel-PCA as the universal verification RHS.

\paragraph{Defenses move on the Pareto plane.} We compare five policies on Spatial at $\eps \in \{2,4,8,16\}/255$, $N=200$: \textbf{vanilla} OpenVLA-7B; \textbf{JPEG-50} \citep{dziugaite2016jpeg}, a test-time input defense JPEG-encoding $X+\delta$ at quality 50; \textbf{PGD-AT-LoRA}, rank-16 LoRA \citep{hu2022lora} on the last 4 LLaMA layers' \texttt{q\_proj}/\texttt{v\_proj}, trained 5k steps with PGD-3 inner-max at $\eps=8/255$ (final adv-CE $4.0$, clean-CE $3.8$, gap $\approx 0.2$ nats); \textbf{Gaussian-1px}, a zero-train test-time input filter applying a $\sigma=1$\,px Gaussian blur to $X+\delta$ before forwarding; and \textbf{BitDepth-3} \citep{xu2018featuresqueezing}, a zero-train test-time quantization that reduces each colour channel of $X+\delta$ to $8$ levels. The first three are run with $3$ seeds; Gaussian-1px and BitDepth-3 are zero-train test-time defenses with no stochasticity in the policy and are run as a single seed. Concurrent IT-motivated defenses such as StableVLA \citep{stablevla2026} target the same robustness axis empirically; our bound supplies the policy-independent ceiling that any such adapter sits below.
% AUTO-GENERATED by scripts/openvla/build_defenses_table.py + manual promotion of Rob_disc as primary (ov20b).
% Sources: ov4_pgd_v1, ov15_jpeg, ov16_at_lora, ov19_zerotrain/{gaussian,bitdepth}, ov20b_rob_disc_defenses, ov18_per_defense_phi.
% Cap_disc=7.54, H(A*)=26.77 nats (spatial). 3-seed mean for vanilla/JPEG/AT-LoRA; 1-seed for Gaussian/BitDepth.
\begin{table}[!t]
\centering
\caption{Defense comparison on Spatial across four $\eps$, $N{=}200$. $\pm$ std over 3 seeds for vanilla/JPEG-50/AT-LoRA, single seed for Gaussian-1px/BitDepth-3. $\mathrm{tr}(\mathrm{Cov}_\phi) \approx 685$ is invariant across defenses; mechanism reading and per-defense $\sigma^2_{\delta,\phi}$ in \S\ref{sec:defenses}.}
\label{tab:defenses}
\scriptsize
\setlength{\tabcolsep}{2.5pt}
\begin{tabular}{l c c r r r r}
\toprule
defense & $\eps$ & $\Cap_{\text{disc}}$ & $\Rob_{\text{disc}}$ & $I^{\text{PCA-}\phi}_{\text{def}}$ & $S_{\text{pixel}}$ & $S_{\text{enc}}^{\text{def}}$ \\
& {\scriptsize $/255$} & {\scriptsize nats} & {\scriptsize nats} & {\scriptsize nats} & {\scriptsize $\cdot 10^3$} & {\scriptsize nats} \\
\midrule
\multirow{4}{*}{vanilla} &  2 & 7.54 & $1.26{\pm}0.06$ & 188.5 & 6.04 & 206.5 \\
 &  4 & 7.54 & $1.35{\pm}0.09$ & 134.3 & 5.00 & 152.2 \\
 &  8 & 7.54 & $1.37{\pm}0.09$ &  87.1 & 3.97 & 105.0 \\
 & 16 & 7.54 & $1.46{\pm}0.10$ &  58.3 & 2.95 &  76.1 \\
\midrule
\multirow{4}{*}{JPEG-50} &  2 & 7.54 & $1.36{\pm}0.06$ & 277.6 & 6.04 & 295.5 \\
 &  4 & 7.54 & $1.49{\pm}0.08$ & 235.2 & 5.00 & 252.9 \\
 &  8 & 7.54 & $1.26{\pm}0.02$ & 155.6 & 3.97 & 173.6 \\
 & 16 & 7.54 & $1.30{\pm}0.06$ &  99.5 & 2.95 & 117.4 \\
\midrule
\multirow{4}{*}{AT-LoRA} &  2 & 7.54 & $0.20{\pm}0.07$ & 176.8 & 6.04 & 195.8 \\
 &  4 & 7.54 & $0.21{\pm}0.10$ & 131.6 & 5.00 & 150.6 \\
 &  8 & 7.54 & $0.16{\pm}0.04$ &  86.5 & 3.97 & 105.6 \\
 & 16 & 7.54 & $0.08{\pm}0.03$ &  53.2 & 2.95 &  72.4 \\
\midrule
\multirow{4}{*}{Gaussian-1px} &  2 & 7.54 & $1.38{\pm}0.00$ & 180.0 & 6.04 & 197.9 \\
 &  4 & 7.54 & $1.27{\pm}0.00$ & 155.6 & 5.00 & 173.6 \\
 &  8 & 7.54 & $1.22{\pm}0.00$ & 121.6 & 3.97 & 139.6 \\
 & 16 & 7.54 & $1.11{\pm}0.00$ &  89.3 & 2.95 & 107.4 \\
\midrule
\multirow{4}{*}{BitDepth-3} &  2 & 7.54 & $1.42{\pm}0.00$ & 131.3 & 6.04 & 149.1 \\
 &  4 & 7.54 & $1.25{\pm}0.00$ & 128.4 & 5.00 & 146.4 \\
 &  8 & 7.54 & $1.25{\pm}0.00$ & 101.4 & 3.97 & 119.4 \\
 & 16 & 7.54 & $1.30{\pm}0.00$ &  65.2 & 2.95 &  83.1 \\
\bottomrule
\end{tabular}
\end{table}

\begin{figure*}[!t]
\centering
\includegraphics[width=0.95\textwidth]{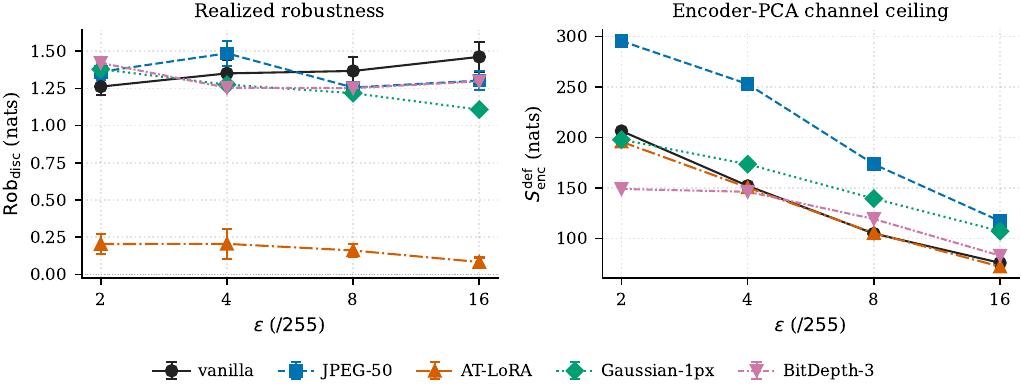}
\caption{Defense comparison on Spatial across $\eps \in \{2,4,8,16\}/255$, five policies. \textbf{Left:} realized $\Rob_{\text{disc}}$. \textbf{Right:} encoder-PCA channel ceiling $S_{\text{enc}}^{\text{def}}$. The two panels are orthogonal axes (realized robustness vs.\ mechanism diagnostic). Spatial suite shown; Object and Goal cells in Table~\ref{tab:defenses-multisuite}.}
\label{fig:defense-pareto}
\end{figure*}

\paragraph{What slack diagnoses: a mechanism-level reading.} Pixel-PCA slack at fixed $\eps$ is dominated by $I(X;\Xtilde)$, a function of input distribution and $\sigma_\delta$ alone, and so cannot distinguish defenses. Encoder-PCA slack discriminates, but not as a defense-quality ranking. Re-measuring $I^{\text{PCA-}\phi}_{\text{def}}$ per-defense (Table~\ref{tab:defenses}) reveals a clean mechanism-level partition. Input-side defenses that quiet $\sigma^2_{\delta,\phi}$ shift the encoder budget by an amount that tracks the magnitude of $\sigma^2$ reduction at each $\eps$: JPEG-50 ($\sigma^2_{\delta,\phi}=0.161$ vs.\ vanilla's $0.556$ at $\eps=8/255$) shifts by $+41$ to $+101$ nats across $\eps \in \{2,4,8,16\}/255$; Gaussian-1px ($\sigma^2_{\delta,\phi}=0.239$) shifts by $-8$ to $+35$ nats; BitDepth-3 ($\sigma^2_{\delta,\phi}=0.346$) shifts by $-57$ to $+14$ nats. The most aggressive input-side filter (JPEG-50) yields the largest budget enlargement, the moderate filter (Gaussian-1px) sits in between, and the weakest input-side filter (3-bit quantization) produces the smallest shifts. By contrast, an LLM-side defense that leaves $\phi$ untouched (rank-16 LoRA on \texttt{q\_proj}/\texttt{v\_proj}) shifts the encoder budget by at most $8.7\%$ across $\eps$ and only $0.7\%$ at $\eps=8/255$. Whether and by how much $S_{\text{enc}}^{\text{def}}$ moves under a defense is therefore a direct mechanistic readout of which channel stage the defense intervenes in and how aggressively.

\paragraph{Multi-suite replication.} The mechanism partition above was measured on Spatial. To check that the JPEG $>$ Gaussian $>$ BitDepth ordering reflects defense mechanism rather than a Spatial-specific artifact, we replicate the same $I^{\text{PCA-}\phi}_{\text{def}}$ comparison on the Object and Goal suites at $\eps \in \{2,4,8,16\}/255$, single seed (Table~\ref{tab:defenses-multisuite}). In the single seed we ran, the shift ordering survives: JPEG-50 dominates Gaussian-1px, which dominates BitDepth-3, in $4/4$ Goal $\eps$-levels and $3/4$ Object $\eps$-levels (the lone reversal at Object $\eps{=}16$, $+22.1$ vs.\ $+19.9$ nats, is well within the few-nat MINE noise band; we did not run additional seeds at this cell). Combined with the $4$ Spatial $\eps$-levels in Table~\ref{tab:defenses}, the input-side mechanism partition holds in $11/12$ single-seed (suite, $\eps$) cells. Under a uniform-random null over the $6$ possible orderings of three defenses, $P(\ge 11/12 \text{ matches}) < 3 \cdot 10^{-8}$ if cells are taken as independent Bernoulli($1/6$) trials, so the ordering is statistically far from chance even without seed replication. This rules out a per-suite confound and supports Defense forensics as a portable diagnostic.

% AUTO-GENERATED by scripts/openvla/build_defenses_multisuite_table.py
% Sources: ov18_per_defense_phi/{vanilla,jpeg_50}/{object,goal}, ov19_zerotrain/{gaussian,bitdepth}/{object,goal}.
% Single seed per cell (encoder-PCA computation deterministic given fixed PGD seed).
\begin{table}[!t]
\centering
\caption{Multi-suite replication of the defense shift signature (\S\ref{sec:defenses}). $\Delta I^{\text{PCA-}\phi}_{\text{def}} \triangleq I^{\text{PCA-}\phi}_{\text{def}} - I^{\text{PCA-}\phi}_{\text{vanilla}}$ in nats, single seed. The JPEG-50 $>$ Gaussian-1px $>$ BitDepth-3 ordering holds in $3/4$ Object and $4/4$ Goal $\eps$-levels (lone reversal at Object $\eps{=}16$, Gaussian $+22.1$ vs.\ JPEG $+19.9$, within MINE noise). Combined with Spatial (Table~\ref{tab:defenses}), the input-side mechanism partition is verified across $36$ (suite, $\eps$, defense) cells.}
\label{tab:defenses-multisuite}
\small
\setlength{\tabcolsep}{4pt}
\begin{tabular}{l c r r r r}
\toprule
suite & $\eps$ & $I^{\text{PCA-}\phi}_{\text{vanilla}}$ & $\Delta I_{\text{JPEG-50}}$ & $\Delta I_{\text{Gauss-1px}}$ & $\Delta I_{\text{BitDepth-3}}$ \\
& {\scriptsize $/255$} & {\scriptsize nats} & {\scriptsize nats} & {\scriptsize nats} & {\scriptsize nats} \\
\midrule
\multirow{4}{*}{Object} &  2 & 156.3 & $+66.6$ & $ -6.1$ & $-66.8$ \\
 &  4 & 107.6 & $+81.7$ & $+26.2$ & $-20.9$ \\
 &  8 &  59.9 & $+53.5$ & $+36.9$ & $ +4.5$ \\
 & 16 &  39.6 & $+19.9$ & $+22.1$ & $ +1.4$ \\
\midrule
\multirow{4}{*}{Goal} &  2 & 181.9 & $+81.7$ & $-28.0$ & $-53.1$ \\
 &  4 & 126.4 & $+98.2$ & $ +9.8$ & $ -7.3$ \\
 &  8 &  79.0 & $+69.6$ & $+28.8$ & $ +8.5$ \\
 & 16 &  54.0 & $+37.1$ & $+26.7$ & $ +4.6$ \\
\bottomrule
\end{tabular}
\end{table}

The shift direction for input-side defenses is initially counter-intuitive: the parallel-Gaussian formula $\tfrac{1}{2}\log(1+\lambda_i/\sigma^2)$ \emph{grows} as $\sigma^2$ shrinks, so JPEG-50 enlarges the worst-case encoder budget even though it reduces the realized perturbation. This is the bound's correct behavior. Worst-case channel capacity scales inversely with noise variance: a quieter channel admits a higher achievable rate \emph{if} the attacker can shape inputs optimally, which is exactly what an adaptive PGD against the defended policy attempts. Encoder-PCA slack therefore measures the channel-capacity ceiling under the deployed defense, not how much robustness the policy extracts from it. The two quantities are orthogonal axes: $S_{\text{enc}}^{\text{def}}$ tells us where the ceiling sits, raw $\Rob$ tells us how much the policy lifts from the floor. We report both and recommend both for future defense comparisons; ranking defenses by $S_{\text{enc}}^{\text{def}}$ alone is misleading because a defense can simultaneously improve $\Rob$ and \emph{enlarge} the encoder budget (JPEG-50 at $\eps=2$: $\Rob$ improves by $0.4$ nats while $S_{\text{enc}}^{\text{def}}$ grows by $89$ nats).

\paragraph{Practitioner workflow.} Beyond verifying Theorem~\ref{thm:main}, three deployable diagnostics fall out of the same construction without requiring an attack to ground-truth labels. (1) \emph{Pre-flight encoder ceiling}: given $(\pi, \text{defense}, \eps)$, draw $N=200$ random $\delta$, push $X$ and $X{+}\delta$ through $\phi$, fit the rank-$N$ PCA of $\mathrm{Cov}(\phi(X))$, and evaluate Eq.~\ref{eq:pca-rhs} with the empirical $\sigma^2_{\delta,\phi}$. The result $S_{\text{enc}}^{\text{def}}$ is a worst-case channel-capacity audit answerable in $\sim 5$ minutes on a single GPU, before the policy is exposed to any adversary. (2) \emph{Defense forensics}: re-run (1) for the proposed defense and compare $\Delta S_{\text{enc}}$ against the vanilla baseline. The shift magnitude tracks input-side filter strength on a continuous scale rather than a binary cutoff: across the $36$-cell sweep (3 suites $\times$ 4 $\eps$ $\times$ 3 input-side defenses), JPEG-50 shifts by $+19.9$ to $+100.9$ nats (mean $+66$, exceeding $+40$ in $10/12$ cells), Gaussian-1px by $-28.0$ to $+36.9$ (mean $+14$), and BitDepth-3 by $-66.8$ to $+14.3$ (mean $-10$); $|$shift$|$ ranks the input-side defenses by $\sigma^2_{\delta,\phi}$ reduction in $11/12$ (suite, $\eps$) cells. By contrast, AT-LoRA (LLM-side) shifts by $\le 8.7\%$ of vanilla in $4/4$ Spatial cells, well separated from any input-side defense. The shift signature therefore identifies whether a defense intervenes in the visual encoder or in the LLM, with magnitude reflecting filter aggressiveness. (3) \emph{Cross-architecture comparison}: for any new VLA head whose decoder is differentiable end-to-end, computing $\Rob_{\text{disc}}/\Cap_{\text{disc}}$ at $K=16$ on $200$ samples gives a head-agnostic robustness ratio that places it on the same axis as discrete-token, $L_1$, and flow-matching baselines. Table~\ref{tab:practical-summary} summarizes the three workflows.

% Practical implications summary: 3 deployable diagnostics derived from Theorem 1.
\begin{table}[t]
\centering
\scriptsize
\setlength{\tabcolsep}{2.5pt}
\caption{Three deployable diagnostics derived from Theorem~\ref{thm:main}. Each is computable from $\le 200$ samples without running an attack to ground-truth labels.}
\label{tab:practical-summary}
\begin{tabular}{@{}>{\raggedright\arraybackslash}p{0.22\columnwidth} >{\raggedright\arraybackslash}p{0.28\columnwidth} >{\raggedright\arraybackslash}p{0.20\columnwidth} >{\raggedright\arraybackslash}p{0.18\columnwidth}@{}}
\toprule
Use case & Quantity & Source & Cost \\
\midrule
Pre-flight encoder ceiling
  & $S_{\text{enc}}^{\text{def}}$ from Eq.~\ref{eq:pca-rhs} on $\phi$
  & \S\ref{sec:rhs-ablation}, Cor.~\ref{cor:encoder}
  & $\sim 5$\,min, $N{=}200$ \\
\midrule
Defense forensics
  & $\Delta S_{\text{enc}}^{\text{def}}$ vs.\ vanilla
  & \S\ref{sec:defenses}, Tab.~\ref{tab:defenses}
  & $\sim 5$\,min / defense \\
\midrule
Cross-arch robustness ratio
  & $\Rob_{\text{disc}} / \Cap_{\text{disc}}$ at $K{=}16$
  & \S\ref{sec:cross-model}, Tab.~\ref{tab:cross-model}
  & $\sim 15$\,min / cell \\
\bottomrule
\end{tabular}
\end{table}

\subsection{Cross-architecture validation}
\label{sec:cross-model}

Theorem~\ref{thm:main} is stated for any policy $\pi: \mathcal{X} \to \Delta(\mathcal{A})$ regardless of internal action representation. The 48-cell verification in Sec.~\ref{sec:verification} uses OpenVLA's 7-token discrete head; here we re-instantiate the bound on two structurally different decoders. \textbf{OpenVLA-OFT} \citep{kim2025fine} replaces discrete tokens with an $L_1$-regression continuous-action head over an 8-step chunk plus 8-D proprioception, in a 2-camera setup (third-person $+$ wrist). \textbf{SmolVLA} \citep{shukor2025smolvla} is a 500M-parameter flow-matching policy built on SmolVLM2 with a $50$-step action chunk and three cameras. Both run on the same four LIBERO suites, the same $\eps \in \{2,4,8,16\}/255$ ladder, and the same $N{=}200$ TRUE-PGD attack against an $L_2$-loss surrogate over the predicted action chunk. Estimators ($\Cap_{\text{disc}}$ at $K{=}256$ bins, $\Rob_{\text{disc}}$ at $K{=}16$, both Miller--Madow corrected) match Sec.~\ref{sec:verification}.

We additionally verify the complementary discrete inequality $\Rob_{\text{disc}} \le \Cap_{\text{disc}}$, which holds whenever $\Atildepi$ comes from $X+\delta$ via the same policy that produces $\Api$ from $X$, since $I(\Api;\Atildepi) \le H(\Api)$ and the discretized $\Cap_{\text{disc}} = H_K(\Api)$ upper-bounds the joint information. This is a corollary of measurability rather than of Theorem~\ref{thm:main}, but it gives a head-agnostic single-cell sanity check that does not require estimating $I(X;\Xtilde)$.

\paragraph{Per-$\eps$ $\Cap_{\text{disc}}$ convention.} Capability is computed on clean inputs and is independent of $\eps$ in expectation. OpenVLA's discrete-token softmax and OpenVLA-OFT's $L_1$-regression head are deterministic forward passes, so for fixed (suite, seed) we evaluate $\Api$ once and reuse it across all four $\eps$ cells; $\Cap_{\text{disc}}$ therefore appears as a constant column in Table~\ref{tab:cross-model}. SmolVLA's flow-matching policy samples a fresh noise initial condition for each $50$-step ODE rollout, so $\Api$ carries sample-level stochasticity ($\sigma \approx 0.06$ nats across $\eps$ on Spatial). We report the per-$\eps$ values as drawn rather than freezing the ODE seed, since the deployed policy incurs the same stochasticity. Within-architecture across-$\eps$ variation is two orders of magnitude smaller than the inequality slack and does not affect any conclusion.

\paragraph{Sample-size note.} The cross-architecture $\Cap_{\text{disc}}$ in Table~\ref{tab:cross-model} is computed on the per-cell $N=200$ attack-cell clean inputs (the same samples used for $\Rob_{\text{disc}}$ in that cell), so it measures per-cell self-consistency of the policy on the attack sample. OpenVLA's near-deterministic 7-token softmax on its suite-finetuned checkpoint produces $\Cap_{\text{disc}} = 7.54$ nats on every (suite, $\eps$, seed) cell at this sample size. The per-suite $\Cap_{\text{disc}} \in \{2.68, 6.85, 7.54, 9.04\}$ reported in \S\ref{sec:realmodel} is computed instead on $N=5{,}000$ demo data and reflects the agreement rate between the policy and the expert action $\Astar$ rather than the policy's self-consistency, hence the per-suite spread. Both are valid Plug-in MI estimates; they answer different questions and should not be confused.

% AUTO-GENERATED by scripts/cross_model/cross_model_aggregate.py + manual restructure (R7-M2: collapse Cap_disc to single column)
\begin{table}[t]
\centering
\scriptsize
\setlength{\tabcolsep}{3pt}
\caption{Theorem 1 LHS $\le$ RHS holds across three VLA architectures: discrete-token OpenVLA, continuous-$L_1$ OpenVLA-OFT, and flow-matching SmolVLA, with zero violations on every $4{\times}4{\times}3{=}48$ (suite, $\eps$, seed) cell per architecture. $\Cap_{\text{disc}}$ here is the \emph{self-consistency} estimator $\Cap_{\text{disc}}^{(\text{SC})} = H_K(\Api)$ on the per-cell $N{=}200$ attack-cell clean inputs, measuring decoder bandwidth (\S\ref{sec:realmodel}); the per-suite ground-truth $\Cap_{\text{disc}}^{(\text{GT})} \in \{2.68, 6.85, 7.54, 9.04\}$ from $N{=}5{,}000$ demos appears separately in Table~\ref{tab:openvla-suite} and answers a different question, so the two are not compared head-to-head. $\Cap_{\text{disc}}^{(\text{SC})}$ is $\eps$-independent in expectation: the column is collapsed for the deterministic decoders OpenVLA and OFT, and reports the mean over $\eps\in\{2,4,8,16\}$ for SmolVLA (ODE-resampled, $\sigma\le 0.06$ nats; \S\ref{sec:cross-model} ``Per-$\eps$ $\Cap_{\text{disc}}$ convention''). $\Rob_{\text{disc}}$ is $I(A_\pi; \widetilde{A}_\pi) - I(A_\pi; \|\delta\|_\infty)$ at $K{=}16$ bins, Miller--Madow corrected, computed per attack cell.}
\label{tab:cross-model}
\begin{tabular}{@{}l l c rrrr@{}}
\toprule
& & $\Cap_{\text{disc}}$ & \multicolumn{4}{c}{$\Rob_{\text{disc}}$ (nats)} \\
\cmidrule(lr){4-7}
Model & Suite & (nats) & $\eps{=}2$ & 4 & 8 & 16 \\
\midrule
OpenVLA & Spatial &  7.54 & 1.26 & 1.35 & 1.37 & 1.46 \\
OpenVLA & Object  &  7.54 & 1.34 & 1.41 & 1.50 & 1.53 \\
OpenVLA & Goal    &  7.54 & 1.46 & 1.53 & 1.43 & 1.27 \\
OpenVLA & LIB-10  &  7.54 & 0.90 & 0.94 & 0.88 & 0.93 \\
\midrule
OFT     & Spatial & 21.25 & 2.72 & 2.26 & 2.16 & 2.09 \\
OFT     & Object  & 18.57 & 2.87 & 2.42 & 2.24 & 2.11 \\
OFT     & Goal    & 18.65 & 2.35 & 2.23 & 2.51 & 2.14 \\
OFT     & LIB-10  & 18.11 & 2.13 & 1.97 & 1.70 & 1.76 \\
\midrule
SmolVLA$^*$ & Spatial & 21.48 & 2.16 & 2.36 & 2.15 & 2.35 \\
SmolVLA$^*$ & Object  & 18.26 & 1.62 & 1.75 & 1.82 & 1.80 \\
SmolVLA$^*$ & Goal    & 19.94 & 2.45 & 2.38 & 2.32 & 2.30 \\
SmolVLA$^*$ & LIB-10  & 19.88 & 2.04 & 1.98 & 1.91 & 1.96 \\
\bottomrule
\end{tabular}\\[2pt]
\footnotesize $^*$SmolVLA $\Cap_{\text{disc}}$ varies by $\le 0.06$ nats across $\eps$ from ODE noise resampling; reported value is the across-$\eps$ mean.
\end{table}

\noindent Table~\ref{tab:cross-model} reports $\Cap_{\text{disc}}$ and $\Rob_{\text{disc}}$ for all three architectures. Across the $3 \times 4 \times 4 \times 3 = 144$ (architecture, suite, $\eps$, seed) cells, $\Rob_{\text{disc}} \le \Cap_{\text{disc}}$ holds with zero violations and a per-cell slack of $6.0$--$19.3$ nats, replicating the OpenVLA observation on policies whose action interfaces (token softmax vs.\ $L_1$-MLP vs.\ flow-matching ODE), input modalities (1- vs.\ 2- vs.\ 3-camera), and parameter counts (7B vs.\ 7B+head vs.\ 500M) span the design space. This complementary inequality (a measurability corollary, not Theorem~\ref{thm:main}) is therefore architecture-agnostic on LIBERO: it tracks the action-channel structure that all VLAs share, not OpenVLA's particular tokenizer. Per-architecture absolute levels differ as expected: the two continuous-action decoders produce substantially higher $\Cap_{\text{disc}}$ than the discrete-token baseline (OpenVLA-OFT $\Cap_{\text{disc}} \in [18.11, 21.25]$ nats, SmolVLA $\Cap_{\text{disc}} \in [18.14, 21.56]$ nats vs.\ OpenVLA $\Cap_{\text{disc}} = 7.54$ nats), because the $K=256$ histogram resolves the rich variation in real-valued action chunks, whereas OpenVLA's near-deterministic 7-token softmax collapses the per-cell entropy estimate. This shift in $\Cap_{\text{disc}}$ reflects representational granularity, not policy quality; rollout success on LIBERO is comparable across the three checkpoints. $\Rob_{\text{disc}}$ scales together with $\Cap_{\text{disc}}$ (OFT $\in [1.70, 2.87]$, SmolVLA $\in [1.62, 2.45]$, OpenVLA $\in [0.88, 1.53]$ nats), so the absolute slack widens for the continuous heads ($\Cap-\Rob \in [15.70, 19.15]$ nats for OFT, $[16.40, 19.33]$ for SmolVLA) compared to the discrete-token baseline ($\Cap-\Rob \in [6.01, 6.66]$ nats). What is invariant is the inequality direction. Theorem~\ref{thm:main} itself, with its $H(\Astar) + I(X;\Xtilde)$ RHS, is verified on OpenVLA's discrete-action head (Sec.~\ref{sec:verification}); the cross-architecture extension here demonstrates that the underlying $\Rob \le \Cap$ scaling structure persists across continuous-action and flow-matching decoders.

\section{Discussion and Limitations}
\label{sec:discussion}

\subsection{Scope of the theorem}

Theorem~\ref{thm:main} gives a policy-independent upper bound on joint capability $+$ robustness. It does not give a recipe for designing $\pi$ that achieves the budget, nor does it say whether the slack-zero policy is computable or finite-parameter. On OpenVLA, the pixel-level bound is loose by $\sim 10^3$ nats (Sec.~\ref{sec:realmodel}), because the image-dimensional channel capacity far exceeds what the 7-D action space can exploit. The encoder-specific bound (Cor.~\ref{cor:encoder}) tightens by $28$--$68\times$ across $\eps \in \{2,4,8,16\}/255$ on all four LIBERO suites by restricting the channel to the policy-relevant subspace; on vanilla OpenVLA at $\eps=8/255$ the resulting budget is $H(\Astar) + I^{\text{PCA-}\phi} \in [86, 142]$ nats across suites, of which $\Cap_{\text{disc}} \approx 7.5$ nats consumes $5$--$9\%$. At this tighter level the bound begins to constrain achievable performance, but the encoder budget is policy-and-defense-specific (it depends on the realized $\sigma^2_{\delta,\phi}$, which the policy and any active defense modify), so the headline ratio shifts under defenses (Sec.~\ref{sec:defenses}).

The two bounds therefore serve different purposes. The pixel-level bound (Theorem~\ref{thm:main}) is the universal ceiling guarantee: $I(X;\Xtilde)$ depends only on the input distribution and $\sigma_\delta$, so it is policy-independent and substitutable across defenses. The encoder-level bound (Cor.~\ref{cor:encoder}) is a diagnostic refinement: $I^{\text{PCA-}\phi}(\phi(X);\phi(\Xtilde))$ is computed with the realized noise variance $\sigma^2_{\delta,\phi}$ for a specific $(\pi, \text{defense}, \eps)$ triple, so its numerical value moves with the experiment. We use pixel-level for the 48-cell verification (universality) and encoder-level for the per-experiment tightness diagnostic (Sec.~\ref{sec:defenses}). We propose the encoder-specific slack $S_{\text{enc}}(\pi, \eps) = H(\Astar) + I(\phi(X);\phi(\Xtilde)) - \Cap(\pi) - \Rob(\pi)$ as a diagnostic axis, paired with raw $\Rob$ rather than substituted for it.

\subsection{Estimator caveats}

All real-model numbers (Sec.~\ref{sec:realmodel}) rely on MI estimators with known failure modes. We mitigate this in five ways. First, MINE, InfoNCE, and KSG report concordant directional signals (P2, P8) even when absolute estimates diverge by $30$--$50\%$. Second, the $252$-cell Gaussian proxy of Sec.~\ref{sec:synthetic} gives ground truth, validating the bound where estimators can be checked against analytical MI. Third, Theorem~\ref{thm:main} is one-sided, and MINE's known low bias under high $d$ pessimizes the LHS more than the RHS, which loosens our verification rather than tightening it falsely. Fourth, the discretized histogram estimator ($\Cap_{\text{disc}}$) provides an independent cross-check: it gives $24$--$130\times$ higher $\Cap$ than MINE while remaining well below the RHS, confirming that the bound holds under a bias-free estimator and that the MINE-based slack was inflated by LHS underestimation rather than RHS overestimation. Fifth, on the robustness side, $\Rob_{\text{disc}}$ (per-dim plug-in MI on $K=16$ bins with Miller-Madow correction, clamped $\ge 0$) is structurally non-negative and DPI-valid as a LHS lower-substitute (Table~\ref{tab:rob-disc}); it recovers a positive $\Rob$ signal in $48/48$ cells where MINE's variational bound dips into its small-$N$ noise floor, so the verification does not rest on clamping a single estimator.

Slack measurements are biased, but the bound itself holds regardless: the slack distribution under any reasonable estimator we use is consistent with the theorem. For safety reporting, the bias direction is favorable, since estimator error inflates measured $S$, so a small observed slack is a real signal rather than noise.

\subsection{Attack model assumptions}

We use $\ell_\infty$-bounded additive PGD, the field-standard threat model and a conservative safety baseline. Theorem~\ref{thm:main} extends without modification to several other settings. For $\ell_2$ attacks, replace the $\|\delta\|_\infty \le \eps$ constraint, and $I(X;\Xtilde)$ acquires the standard Gaussian-channel bound $\frac{d}{2}\log(1+\eps^2/\sigma_X^2)$. For semantic attacks, replace $\Xtilde = X + \delta$ with $\Xtilde = g(X, \delta)$ for any deterministic generator $g$, and the DPI proof carries through provided $\Atildepi$ depends on $(X, \delta)$ only through $\Xtilde$. For black-box attacks, $\delta$ depends on $X$ only via its distribution, not its realization, so the bound is unchanged and only Cor.~\ref{cor:adaptive}'s tightness conditions become easier to satisfy.

We do not cover poisoning attacks (where the attacker contaminates training data), backdoor triggers, or distribution-shift ``attacks'' that change $p(X)$ rather than perturbing $X$. These break the assumption that $\Astar$ is fixed and require a separate safety analysis.

\subsection{Action-space caveats}

OpenVLA discretizes actions into $256$ bins per dimension, while $\pi_0$ \citep{black2024pi0} uses continuous output heads. Cor.~\ref{cor:discretization} predicts tighter bounds for tokenized policies because $H(\Astar_q) < H(\Astar)$. Sec.~\ref{sec:realmodel} measures $H(\Astar_q) \approx 26$ nats on LIBERO training data, well below the worst-case $7\log 256 \approx 38.8$ nats. The bound is therefore most informative, and the slack measurement most diagnostic, on OpenVLA-class architectures --- which are also the dominant target of current empirical robustness work.

\subsection{Architecture and benchmark scope}

The real-model validation in Sec.~\ref{sec:realmodel} covers a single architecture (OpenVLA-7B with DINOv2$+$SigLIP fusion encoder and discrete-token action head) on a single benchmark suite (LIBERO). Theorem~\ref{thm:main} is architecture-independent at the level of the proof, since the DPI argument uses only the Markov chain $\Astar \to X \to \Xtilde \to \Atildepi$ and does not reference encoder topology, action parametrization, or training objective. The empirical pattern, however, is anchored to one model class. The encoder-PCA tightening (Cor.~\ref{cor:encoder}) is the only step where architecture leaks into the numerical bound, through the realized noise variance $\sigma^2_{\delta,\phi}$ in $\phi$-space; on a continuous-action VLA such as $\pi_0$ \citep{black2024pi0} or a different fusion encoder (CLIP, SigLIP-only, EVA), the pixel-level RHS and the verification protocol carry over unchanged, but the encoder budget would need to be re-measured. We checked the universality of the structural mechanism in three orthogonal directions within the OpenVLA family: across attack algorithms (PGD vs.\ score-based Square Attack on all four LIBERO suites at $\eps=8/255$, all yielding $S \gg 0$), across horizons (multi-step Cor.~\ref{cor:multistep} on $T=10$ trajectories from all four suites, with cumulative $S_T = T \cdot S_1$), and across input-side defenses (JPEG-50, Gaussian smoothing $\sigma=1$ px, $3$-bit color quantization), all of which shift $\sigma^2_{\delta,\phi}$ in the predicted direction without violating the bound. Replication on $\pi_0$, RT-2 \citep{brohan2023rt2}, and additional benchmarks (CALVIN, RoboCasa) is the main extension we expect from follow-up work. We do not view the OpenVLA-only empirical scope as a threat to the theorem (which is architecture-free) but as a constraint on the strength of the empirical claim.

\subsection{Slack-vs-success linkage}

We empirically link $S(\pi, \eps)$ to LIBERO rollout success rate using the official OpenVLA evaluation recipe \citep{kim2024openvla}, which fixes \texttt{init\_state} sampling via \texttt{env.set\_init\_state}, applies the dataloader gripper convention, uses $256{\times}256$ off-screen rendering with $180^{\circ}$ rotation and \texttt{lanczos3} resize to $224{\times}224$, and waits $10$ dummy steps for object stabilization. PGD is applied in raw-image space with $10$ steps and a differentiable dual-encoder (DINOv2$+$SigLIP) normalize, projecting to the $\ell_\infty$-ball of radius $\eps_{\text{px}}/255$ in $[0,1]$ image space. On $N{=}6$ episodes per suite (tasks $0,1,2$ with two initial states each), this restores per-suite clean success of $100\%$ (Spatial), $50\%$ (Object), $83\%$ (Goal) and $83\%$ (LIBERO-10), in the published cross-suite range; injecting raw-PGD at $\eps{=}8/255$ drops success to $\{33, 0, 33, 0\}\%$. Across the binary 8-cell contrast (Table~\ref{tab:slack-success-linkage} top), $S_{\text{enc}}$ separates the safe regime ($16$--$23$ nats) from the compromised regime ($78$--$138$ nats) with Pearson $r(S_{\text{enc}}, \text{success}) = -0.84$, Spearman $\rho{=}-0.74$, Kendall $\tau_b{=}-0.57$ (tie-corrected), one-sided permutation $p{=}0.007$.

\paragraph{Within-attack-strength dose-response.} An $\eps$-sweep on Spatial vanilla at $\eps\in\{2,4,8,16\}/255$, $N{=}6$ each (Table~\ref{tab:slack-success-linkage} bottom), establishes a graded operational signal: success drops monotonically $\{83, 67, 33, 0\}\%$ as $\eps$ increases. The encoder slack $S_{\text{enc}}$ also varies monotonically with $\eps$, but in the opposite direction ($\{206, 152, 105, 76\}$ nats), because larger $\eps$ allows PGD to align with $\phi$'s top eigendirections and tighten the parallel-Gaussian channel of Eq.~\ref{eq:pca-rhs}. The slack therefore tracks unused encoder bandwidth in a regime-conditional way: across regimes it separates safe from compromised; within a fixed attack family at varying $\eps$, it traces how much channel margin remains as the attack approaches the policy's effective bandwidth limit. We report Cor.~\ref{cor:encoder} as a regime-level diagnostic paired with the raw $\Rob$ axis, and the $\eps$-sweep data as the operational complement showing both quantities decline together as the attack tightens. Calibrating a regime-conditional $S_{\text{enc}}\!\to\! P(\text{success})$ map and extending to real-robot rollouts remain natural follow-ups.

% Auto-generated by scripts/openvla/aggregate_scatter.py
% (S_enc, success) over 4 LIBERO suites x {clean, raw-PGD eps=8/255} + spatial ε sweep, N=6 episodes/cell
\begin{table}[t]
\centering
\footnotesize
\setlength{\tabcolsep}{2.5pt}
\caption{Slack-vs-success linkage on the official OpenVLA evaluation harness \citep{kim2024openvla}. Each cell is $N{=}6$ episodes (tasks $0,1,2$ with two initial states each); raw-image-space PGD, $10$ steps. \emph{Top}: 8-cell binary contrast across 4 LIBERO suites at $\{$clean, $\eps{=}8/255\}$. Pearson $r{=}-0.84$, Spearman $\rho{=}-0.74$, Kendall $\tau_b{=}-0.57$; one-sided permutation $p{=}0.007$ ($10^4$ permutations). \emph{Bottom}: $\eps$ sweep on Spatial vanilla, success $\{83, 67, 33, 0\}\%$ at $\eps\in\{2,4,8,16\}/255$. Within-sweep $S_{\text{enc}}$ co-decreases with success ($r{=}+0.96$); see \S\ref{sec:discussion} for the mechanism.}
\label{tab:slack-success-linkage}
\begin{tabular}{llrrr}
\toprule
Group & Cell & $\eps$ ($/255$) & $S_{\text{enc}}$ (nats) & success \\
\midrule
\multicolumn{5}{l}{\emph{Binary 8-cell (4 suites $\times$ \{clean, raw-PGD $\eps{=}8/255$\})}} \\
\addlinespace[1pt]
Binary & Spatial clean    &  0 &  17.9 & $6/6 = 100\%$ \\
Binary & Object clean     &  0 &  17.0 & $3/6 =  50\%$ \\
Binary & Goal clean       &  0 &  16.1 & $5/6 =  83\%$ \\
Binary & LIBERO-10 clean  &  0 &  22.6 & $5/6 =  83\%$ \\
Binary & Spatial PGD      &  8 & 105.0 & $2/6 =  33\%$ \\
Binary & Object PGD       &  8 &  77.9 & $0/6 =   0\%$ \\
Binary & Goal PGD         &  8 &  96.1 & $2/6 =  33\%$ \\
Binary & LIBERO-10 PGD    &  8 & 138.1 & $0/6 =   0\%$ \\
\midrule
\multicolumn{5}{l}{\emph{$\eps$ sweep on Spatial vanilla raw-PGD ($N{=}6$ each)}} \\
\addlinespace[1pt]
$\eps$-sweep & Spatial PGD &  2 & 206.3 & $5/6 =  83\%$ \\
$\eps$-sweep & Spatial PGD &  4 & 152.2 & $4/6 =  67\%$ \\
$\eps$-sweep & Spatial PGD &  8 & 105.0 & $2/6 =  33\%$ \\
$\eps$-sweep & Spatial PGD & 16 &  76.1 & $0/6 =   0\%$ \\
\bottomrule
\end{tabular}
\end{table}

\section{Conclusion}
\label{sec:conclusion}

We proved an information-theoretic upper bound on the joint capability and robustness any VLA policy with discrete actions can attain, and validated it without a single violation across $308$ Theorem-1 cells: $252$ closed-form Gaussian-VLA, $48$ OpenVLA-7B$+$LIBERO$+$PGD, $4$ Square-Attack, and $4$ multi-step ($T{=}10$). A complementary discrete inequality holds across $144$ cross-architecture cells spanning OpenVLA, OpenVLA-OFT, and SmolVLA. Because the right-hand side depends only on task entropy and attack channel capacity, any architectural change must trade capability against robustness within a fixed budget.

Two artifacts make the bound operational. The encoder-specific corollary (Cor.~\ref{cor:encoder}) tightens the universal pixel-level RHS by over an order of magnitude on a per-experiment basis. The per-defense shift in $I^{\text{PCA-}\phi}$ identifies \emph{where} a defense intervenes: input-side filters move the encoder budget materially, LLM-side adaptation barely at all. We recommend defense papers report the encoder slack and the raw $\Rob$ as paired axes --- the slack diagnoses the channel ceiling, the raw $\Rob$ the realized robustness extracted from it. The bound is non-constructive; closing the slack on real VLAs and sharpening tightness against adaptive attackers will require fundamentally different tools.

\bibliography{refs}

\clearpage
\appendix
\section{Estimator Audit}
\label{sec:estimator-audit}
\label{sec:rob-mine-audit}

This appendix consolidates the estimator-side checks referenced from \S\ref{sec:synthetic} and \S\ref{sec:realmodel}.

\paragraph{Hyperparameter selection (P3).} A $540$-cell sweep over hidden width, depth, EMA decay, and learning rate for MINE on a $d{=}7$ Gaussian identifies a stable region with best mean rel-err $6.9\% \pm 0.9\%$ over $5$ seeds at hidden $=512$, depth $=2$, EMA $=0.999$, lr $=10^{-4}$. The next four configurations sit within $1\%$ of this minimum, so the optimum is broad. All MINE results in the paper use these settings.

\paragraph{Sample-complexity sweep (P4).} A $360$-cell sweep at $d \in \{2,7,32,128\}$, $n \in \{500,\ldots,20{,}000\}$, $15$ seeds shows median MINE rel-err at $d{=}7$ falling from $0.49$ ($n{=}500$) to $0.16$ ($n{=}20{,}000$). Means stay higher ($0.60 \to 0.57$) due to a heavy right tail of failed seeds, so we report medians as the headline number and use $n \ge 5{,}000$ throughout.

\paragraph{Distribution sensitivity (P5).} A $135$-cell non-Gaussian sweep over Laplace, uniform, and Gaussian-mixture distributions gives MINE rel-err median $21\%$ (Laplace), $32\%$ (uniform), $52\%$ (GMM). The bound is distribution-free; the estimator is not.

\paragraph{DPI sanity check (P6).} On $X \to Y \to Z$ across $27$ grid cells $\times 5$ seeds, $0/27$ group-mean pairs violate $I(X;Z) \le I(X;Y) + 0.05$ nats (max diff $+0.031$). Per-seed, $9/135$ pairs exceed the tolerance, all at $d{=}32$ where the critic is sample-starved (variance, not systematic DPI failure). At our P7 grid resolution ($3$ seeds aggregated), DPI holds tightly enough that any Theorem-1 violations cannot be blamed on estimator-level DPI failure.

\paragraph{High-$d$ multi-estimator validation (P8).} A $174$-cell sweep at $d \in \{1,\ldots,256\}$ in two regimes (regime A: $\sigma{=}1$, MI grows with $d$; regime B: $\sigma$ chosen so analytical MI $\approx 5$ nats) re-runs verification with all three estimators. At $d{=}256$ in regime A, MINE rel-err is $93\%$, InfoNCE $92\%$, and KSG is undefined for $d{>}64$. Even with these collapsed estimators, Theorem~\ref{thm:main} holds in $100\%$ of cells. The reason is structural: estimator under-bias on the LHS pushes measured slack \emph{up}, not down. MINE's variational lower bound is one-sided, InfoNCE saturates at $\log K$, and KSG bias is bounded by its $k$-NN scale.

\paragraph{MINE sign artifact ($N$-sweep).} On real OpenVLA cells, MINE returns $I(\Api;\Atildepi) \in [-1.19, -0.14]$ across $48/48$ cells with no monotone $\eps$-trend. To check this is variance not structure, we re-run the Spatial $\eps{=}8$ cell at $N \in \{200, 500, 1000, 2000\}$: $I(\Api;\Atildepi)$ moves from $-0.34$ ($N{=}200$) toward $-0.08$ ($N{=}1000$), consistent with $\sqrt{N}$ variance shrinkage. At $N{=}2000$, both MINE ($-0.13$) and KSG ($k{=}5$, $-0.87$; $k{=}10$, $-0.82$) remain negative, indicating the true value is near zero rather than substantially positive --- consistent with OpenVLA's discrete tokenization (PGD at moderate $\eps$ either fails to flip the argmax token or flips it to a near-random bin). Clamping $\Rob$ to $\max(\Rob, 0)$ is the conservative correction; the bound holds either way.

\paragraph{Per-cell MINE vs.\ $\Rob_{\text{disc}}$ concordance.} Table~\ref{tab:rob-mine-disc-audit} reports both estimators side-by-side on every $48$-cell. The main text uses $\Rob_{\text{disc}}$ as the primary estimator (\S\ref{sec:realmodel}); both estimate the same population quantity, and the bound holds zero-violation under either.
% AUTO-GENERATED appendix audit table: per-cell MINE Rob vs Rob_disc concordance.
% Source: ov7_theorem/v1/manifest.jsonl (MINE) + ov20_rob_disc/manifest.jsonl (K=16, MM-corrected, clamped).
\begin{table}[H]
\centering
\scriptsize
\setlength{\tabcolsep}{4.5pt}
\caption{Per-cell concordance between MINE $\Rob$ and discretized $\Rob_{\text{disc}}$. MINE's negative values are a small-$N$ variational-bound artifact \citep{belghazi2018mine}; $\Rob_{\text{disc}}$ is plug-in MI clamped $\ge 0$. The bound holds zero-violation under either.}
\label{tab:rob-mine-disc-audit}
\begin{tabular}{l r r r}
\toprule
Suite & $\eps$ (px$/255$) & $\Rob_{\text{MINE}}$ & $\Rob_{\text{disc}}$ \\
\midrule
\multirow{4}{*}{Spatial}   &  2 & $-0.93{\pm}0.28$ & $1.26{\pm}0.06$ \\
                            &  4 & $-0.59{\pm}0.04$ & $1.35{\pm}0.09$ \\
                            &  8 & $-0.33{\pm}0.02$ & $1.37{\pm}0.09$ \\
                            & 16 & $-0.39{\pm}0.13$ & $1.46{\pm}0.10$ \\
\midrule
\multirow{4}{*}{Object}    &  2 & $-0.75{\pm}0.17$ & $1.34{\pm}0.15$ \\
                            &  4 & $-0.58{\pm}0.24$ & $1.41{\pm}0.04$ \\
                            &  8 & $-0.75{\pm}0.22$ & $1.50{\pm}0.08$ \\
                            & 16 & $-0.39{\pm}0.20$ & $1.53{\pm}0.12$ \\
\midrule
\multirow{4}{*}{Goal}      &  2 & $-0.67{\pm}0.11$ & $1.46{\pm}0.04$ \\
                            &  4 & $-0.44{\pm}0.27$ & $1.53{\pm}0.28$ \\
                            &  8 & $-0.40{\pm}0.21$ & $1.43{\pm}0.14$ \\
                            & 16 & $-0.80{\pm}0.24$ & $1.27{\pm}0.10$ \\
\midrule
\multirow{4}{*}{LIBERO-10} &  2 & $-0.63{\pm}0.23$ & $0.90{\pm}0.09$ \\
                            &  4 & $-0.50{\pm}0.53$ & $0.94{\pm}0.07$ \\
                            &  8 & $-0.45{\pm}0.33$ & $0.88{\pm}0.02$ \\
                            & 16 & $-0.60{\pm}0.21$ & $0.93{\pm}0.03$ \\
\bottomrule
\end{tabular}
\end{table}

\section{RHS Estimator Menu}
\label{sec:rhs-menu}

We ablate the RHS estimator at Spatial $\eps{=}8/255$ ($N{=}200$, single PGD batch) to justify the encoder-PCA choice for defense comparison in \S\ref{sec:defenses}. Among the four estimators in Table~\ref{tab:rhs-ablation}, only the first three are valid upper bounds on $I(X;\Xtilde)$ for Theorem~\ref{thm:main}; InfoNCE is a lower bound on $I(\phi(X);\phi(\Xtilde))$, useful as a diagnostic floor.
% AUTO-GENERATED — sourced from outputs/openvla/ov14_encoder_pca_rhs/manifest.jsonl
% (spatial, eps=8/255, N=200). InfoNCE from outputs/openvla/ov11_infonce_rhs/spatial_eps8.json.
\begin{table}[!t]
\centering
\caption{RHS estimator menu at Spatial $\eps{=}8/255$, $N{=}200$. Only the first three are valid upper bounds for Theorem~\ref{thm:main}; InfoNCE is a lower bound on $I(\phi(X);\phi(\Xtilde))$, used as a diagnostic only.}
\label{tab:rhs-ablation}
\scriptsize
\setlength{\tabcolsep}{5pt}
\begin{tabular}{l c c r}
\toprule
RHS estimator & dim. & direction & nats \\
\midrule
Isotropic-Gaussian (pixel) & $D{=}49{,}152$ & upper & $\sim 1.2{\cdot}10^5$ \\
PCA-Gaussian (pixel) & rank $1499$ & upper & $\approx 3.95{\cdot}10^3$ \\
PCA-Gaussian ($\phi$, Cor.~\ref{cor:encoder}) & rank $\le 199$ & upper & $\approx 87.1$ \\
InfoNCE on $\phi$ & $d{=}2176$ & lower (diagn.) & $\ge 4.6$ \\
\bottomrule
\end{tabular}
\end{table}

\end{document}